\documentclass[reqno]{amsart}

\usepackage[margin=3cm]{geometry}

\usepackage{arxiv_style}
\usepackage{fontawesome}

\makeatletter
\def\blfootnote{\gdef\@thefnmark{}\@footnotetext}
\makeatother

\title[\IOT: A Client-optimal Oblivious Transfer Protocol for IoT Devices]{\IOT: A Client-optimal Oblivious Transfer Protocol for IoT Devices}

\author[E.\ Onofri, A.\ Ciccotelli, and R.\ Di Pietro]{}

\begin{document}

\blfootnote{$^{\star}$ Corresponding Author, (\href{mailto:elia.onofri@kaust.edu.sa}{\faEnvelopeO}) \texttt{elia.onofri@kaust.edu.sa}}

\maketitle

\vspace{-1em}

\begin{center}
    \begin{minipage}{.89\linewidth}\centering
        \textsc{Elia Onofri}$^{\,\star,\, \orcidlink{0000-0001-8391-2563}}$,
        \textsc{Andrea Ciccotelli}$^{\orcidlink{0009-0004-8286-5135}}$,
        \textsc{Roberto Di Pietro}$^{\, \orcidlink{0000-0003-1909-0336}}$
        \\
        \bigskip
        \begin{minipage}{.90\linewidth}\centering
            \footnotesize
            King Abdullah University of Science and Technology (KAUST)\\
            Computer, Electrical and Mathematical Sciences and Engineering (CEMSE) Division,\\
            Thuwal 23955, Saudi Arabia
        \end{minipage}
    \end{minipage}
\end{center}

\medskip
\thispagestyle{empty}

\begin{abstract}
    Oblivious Transfer (OT) is a fundamental cryptographic primitive enabling privacy-preserving computation and constitutes a core building block for secure multi-party computation while supporting a wide range of security-sensitive applications: private information retrieval, zero-knowledge proofs, and password-authenticated key exchange, to cite a few.
    While recent advances in OT extension have significantly reduced amortised costs, their reliance on batches of random base OTs and substantial pre-computation phases limits their practicality in scenarios where the number of transfers is modest or where communication latency and client-side computation are critical constraints.
    In such settings, efficient base OT protocols remain both relevant and necessary.
    In this work, we introduce \IOT, a novel base 1-out-of-2 OT protocol grounded in the quadratic residuosity problem, specifically designed to minimise receiver-side computation and interaction.
    Our construction is particularly appealing on client--server architectures in which the receiver operates on low-power hardware, such as Internet of Things (IoT) devices.
    Through a lightweight offline pre-computation phase, \IOT shifts the on-transfer computational burden almost entirely to the Sender, while reducing online communication to only six messages and four digests exchanged.
    We provide a detailed description of the protocol, accompanied by a formal proof of its security.
    Moreover, to demonstrate the viability of \IOT, we also present an open-source proof-of-concept implementation (in C language) evaluated on real IoT hardware.
    Results are staggering: for 128-bit security using a 3072-bit RSA modulus, the receiver incurs an average online cost per OT as low as 2.80~\textmu{}s on desktop platforms and 39.90~\textmu{}s on IoT devices, more than 10$\times$ faster than the well known \textit{SimplestOT}.
    These results show that \IOT effectively bridges the gap between theoretical security and practical efficiency, enabling OT deployment in real-world resource-constrained environments.

    \medskip



\end{abstract}

\begin{multicols}{2}
    \tableofcontents
\end{multicols}



\section{Introduction}\label{sec:intro}

The security of modern cryptographic systems, particularly in distributed environments, relies on the ability to perform computations while preserving data privacy.
One of the fundamental cryptographic primitives enabling such privacy-preserving computations is Oblivious Transfer (OT).
Originally introduced by Rabin in 1981~\cite{Rabin1981} as a technique to probabilistically transmit encrypted messages, OT protocols quickly became a fundamental building block of secure multi-party computation (MPC), allowing a Sender to transfer one of multiple possible messages to a Receiver without revealing any additional information.
Specifically, in its simplest form of 1-out-of-2 OT (\OT), the Sender possesses two messages ($M_0, M_1$) and the Receiver can select one of the two ($b\in\{0, 1\}$); \OT then ensures that the Receiver learns $M_b$ while remaining oblivious to $M_{1-b}$ and the Sender learns nothing at all about the choice $b$ (see Figure~\ref{fig:ot-protocol}).

Nowadays, OT serves as a critical building block in numerous cryptographic protocols, including private information retrieval (PIR), secure voting, password-authenticated key exchange (PAKE), and zero-knowledge proofs, to name a few. Its ability to enable selective data transfer while ensuring privacy for both parties makes it essential for security-sensitive applications \cite{rajagopalan2024zero,10.1145/3340301.3341128,10292713,10124397,Ghaleb24}.

Despite its theoretical significance, OT also presents notable computational and efficiency challenges in practical applications.
Traditional secure OT schemes often rely on hardness assumptions from number theory, like the Diffie-Hellman assumptions or the lattice-based hardness assumptions, to guarantee security \cite{Rabin1981,tzeng2004efficient,naor2001efficient,10.1007/978-3-030-45727-3_23,asharov2015malicious}. While these cryptographic foundations provide strong security guarantees, they frequently suffer from high computational and communication overhead, making them impractical for large-scale deployments or resource-constrained environments.

Although OT extension schemes~\cite{beaver96,ishai2003extending} appear to render further work on base OT obsolete, they rely fundamentally on a batch of random base OTs and require a substantial pre-computation phase before any extension can be performed. As a consequence, research has largely concentrated on reducing the number of random base OTs, rather than optimising the classical 1-out-of-2 OT itself. In applications where the total number of OTs remains modest, the cost of this initial setup cannot be effectively amortised: on commonly found implementations, random base OTs typically become advantageous on the order of tens to hundreds of thousands invocations, and deriving a standard \OT from a random OT introduces additional computational overhead. For these scenarios, focusing on efficient base OT still remains both relevant and necessary.

Additionally, some existing OT schemes require an initial message exchange phase before data transfer can begin, further increasing latency and inefficiency. This is \eg the case of the state-of-the-art solution for OT extension \cite{canetti2020blazing}, which requires a 3-round interaction for its base random OT. This issue is particularly problematic in client-server architectures where the client operates on low-power hardware, such as Internet of Things (IoT) devices, and where communication channels are unstable or present sensible throttling. Optimising OT protocols to reduce both communication complexity and client-side computation is therefore crucial to expanding their applicability in real-world scenarios.

To fill the above-highlighted gaps, we propose a novel protocol under the name of \IOT, focusing on a highly efficient transfer and Receiver design capable of providing security even in the presence of active adversaries.


\subsection{Technical Contributions}

In this work, we propose a novel OT protocol designed to achieve high computational Receiver efficiency at a reduced communication cost (only two rounds of interaction are needed) while preserving security guarantees.
The solution is particularly suited for client-server architectures, where the server operates on standard hardware while the client is a resource-constrained IoT device.
By minimising client-side operations and further optimising them under a savvy pre-computation phase, we enable OT to be deployed in scenarios where traditional schemes would be computationally prohibitive.

The proposed protocol is grounded in the ring $\mathbb{Z}_n$, where $n=p\cdot q$ as usual, and $p\equiv_4 q \equiv_4 1$ is required. In particular, it leverages: (i) pseudo-randomness generated from a Key Derivation Function (KDF) $F$ to ensure message confidentiality; (ii) the quadratic residuosity problem to enforce security guarantees; and, (iii) a cryptographic hash function $H(\cdot)$ to disambiguate roots in $\ZZ_n$.

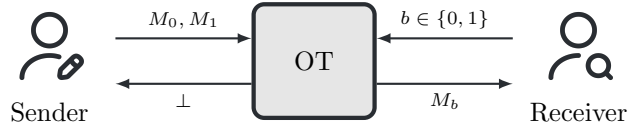
\begin{figure}[t]
    \centering
    \begin{tikzpicture}
        \node[label={below:Sender}] (S) at (0, 0.2) {\includegraphics[width=1cm]{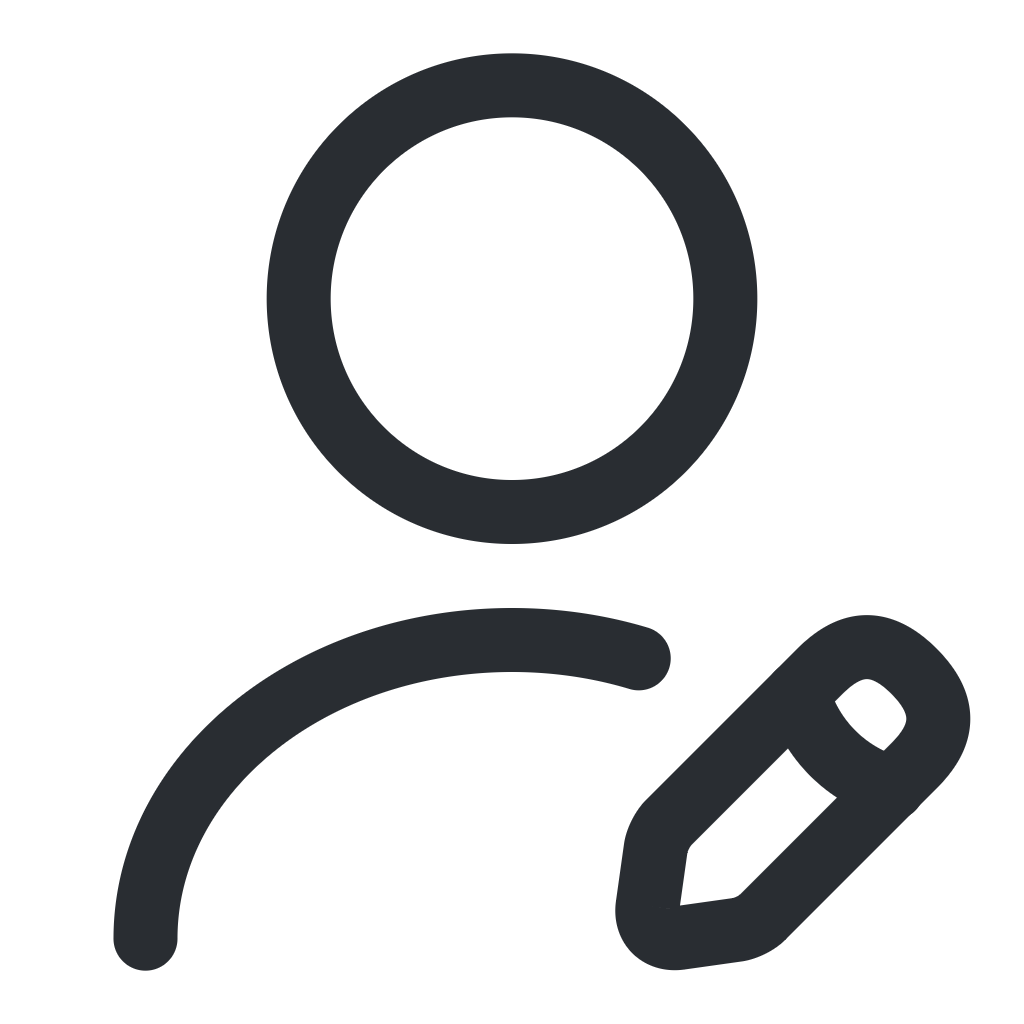}};
        \node[label={below:Receiver}] (R) at (7, 0.2) {\includegraphics[width=1cm]{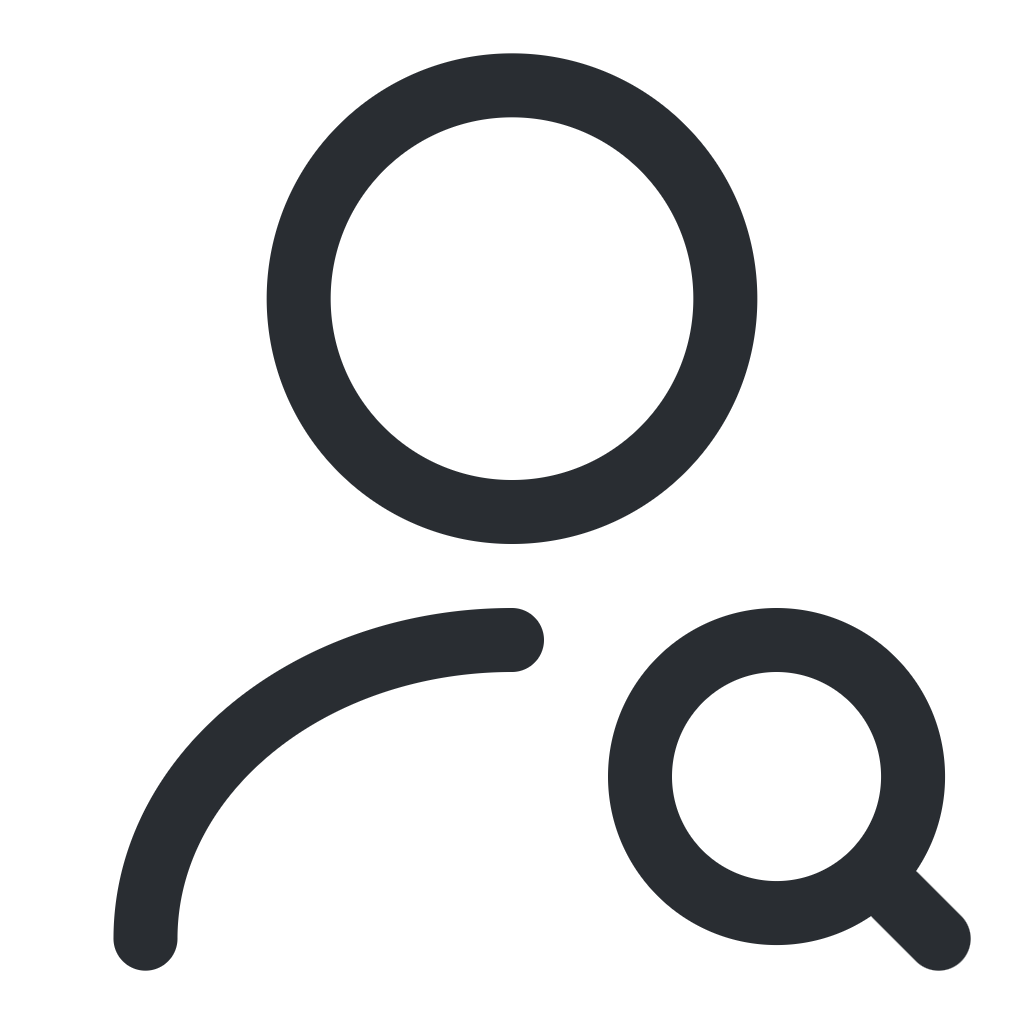}};
        \node[draw=graydraw, fill=black!10, ultra thick, rounded corners, minimum width=1.6cm, minimum height=1.5cm] (OT) at (3.5, 0) {OT};

        \draw[latex-, thick, color=graydraw] ([yshift=.3cm]OT.west) --node[above, black] {\scriptsize $M_0, M_1$} ++(-1.8, 0);
        \draw[-latex, thick, color=graydraw] ([yshift=-.3cm]OT.west) --node[below, black] {\scriptsize $\bot$} ++(-1.8, 0);

        \draw[latex-, thick, color=graydraw] ([yshift=.3cm]OT.east) --node[above, black] {\scriptsize $b \in \{0, 1\}$} ++(1.8, 0);
        \draw[-latex, thick, color=graydraw] ([yshift=-.3cm]OT.east) --node[below, black] {\scriptsize $M_b$} ++(1.8, 0);
    \end{tikzpicture}
    \caption{Abstraction of a Base OT protocol.}
    \label{fig:ot-protocol}
\end{figure}

In detail, \IOT introduces several technical contributions, synthesised in the following:
\begin{description}
	\item[Minimal message exchanges:] The proposed solution eliminates the need for an initial message exchange from the Sender to the Receiver. The Receiver sends a single ring element (masked key) to the Sender, which then responds with a nonce for the KDF (say one ring element to ease exposition), 4 ciphertexts, and four key digests (typically 256 bits). This results in a total of two ring elements, four messages, and four digests, optimising communication overhead.

	\item[Security against a malicious Sender:] Different from other OT proposals, which are designed in the context of honest-but-curious senders, our transfers are secure by design against malicious behaviours.
	Malicious senders can still poison the protocol during setup by modifying the required congruencies on $p$ and $q$; however, we define a probabilistic game that any Receiver can run to validate the modulus trustworthiness. A full probabilistic analysis of said game is also provided and the game framework itself can be extended to different hypothesis testing.

	\item[Optimal Receiver-side overhead:] By requiring only one random ring element sampling, a single call to the Hash function and the KDF, one modular squaring, one message xor, and at most one subtraction, the proposed scheme is, to the best of our knowledge, the most efficient classical OT scheme Receiver-side.
    Furthermore, all the operations, but for the XOR and the KDF evaluation, can be precomputed offline, making it extremely competitive in the online transfer phase despite the primitive computational impact, and hence particularly suitable in the context of IoT.

    \item[Open source code and testing] We realised an open source implementation\footnote{Online available at \url{https://github.com/eOnofri04/IOT2}.} (proof-of-concept) in C language, running the scheme on a single machine to assess its performance. In particular, we benchmark it against both a desktop personal computer and an IoT Raspberry Pi Zero 2W, achieving, to the best of our knowledge, the lowest online computational burden on the Receiver side (more than 10$\times$ faster than the online execution of \cite{CO-OT15}, often used as a benchmark in the literature).
    We also release a python multi-machine implementation to show its feasibility in real contexts.
\end{description}

By combining these innovations, our protocol achieves a lightweight, efficient, and secure OT mechanism, making it particularly well-suited for IoT applications.


\subsection{Paper Organisation}

The remainder of this manuscript is organised as follows.
In Section~\ref{sec:related-work}, we present the OT context and the related literature.
A short introduction of the notation and the mathematical background follows in Section \ref{sec:notation}.
Within Section~\ref{sec:IOT2}, we present our novel approach to OT, or \IOT, formally discussing the system model and its corresponding expected security (Section~\ref{ssec:system-model}), the various threat models (Section~\ref{ssec:threat-models}), the actual scheme (Section~\ref{ssec:scheme}), and the computational complexity of the transfer itself (Section~\ref{ssec:complexity-analysis}).
In Section~\ref{sec:security-analysis}, we then proceed to present the security under the proposed threat models (Section~\ref{ssec:securityAnalysis}, see Appendix~\ref{app:security-games} for a formal definition of the security games involved), which also covers the setting of untrusted senders, as detailed in Section~\ref{ssec:untrustedServer}, where we present an innovative game to assess sender trustworthiness (a detailed discussion on the game's success probability can be found in Appendix~\ref{app:untrusted-sender}).
In Section~\ref{sec:experiments}, we explore the details of our experimental campaign, presenting the settings of our proof-of-concept implementation (Section~\ref{ssec:experimental-setting}) and discussing the corresponding results (Section~\ref{ssec:experimental-results}).
Finally, Section~\ref{sec:conclusions} concludes the paper, summarising the contribution and discussing future work.


\section{Related Work}
\label{sec:related-work}

\begin{table*}[p]

    \def\implnotice{%
        \multicolumn{7}{c}{%
            \raisebox{.1em}{%
                \scalebox{.6}[.8]{%
                    \color{gray}%
                    \rule[0.2em]{4.75cm}{0.5pt}%
                    \phantom{xx}%
                    Not addressed by authors.%
                    \phantom{xx}%
                    \rule[0.2em]{4.75cm}{0.5pt}%
                }%
            }%
        }%
    }

    \def\notprovidednotice{%
        \multicolumn{3}{c|}{%
            \raisebox{.1em}{%
                \scalebox{.6}[.8]{%
                    \color{gray}%
                    \rule[0.2em]{.95cm}{0.5pt}%
                    \phantom{x} Not addressed by authors.%
                    \phantom{x}%
                    \rule[0.2em]{.95cm}{0.5pt}%
                }%
            }%
        }%
    }

    \def\notprovidednoticeshort{%
        \raisebox{.1em}{%
            \scalebox{.6}[.8]{%
                \color{gray}%
                Not addr.%
            }%
        }%
    }

    \centering

    \caption{
        Summary of the structural features and structural statistics of 20 representative OT protocols from the literature, here compared against the proposed \IOT protocol.
    }
    \label{tab:comparison_table}

    \setlength{\tabcolsep}{5pt}
    \def\arraystretch{1.2}
    \rotatebox{90}{\scalebox{.8}{
        \begin{tabular}{cc|ccc|ccc|ccccc|ccc}
        \toprule
        \rowRot{ \textbf{Year} } &

        \rowRot{ \textbf{Protocol} } &
        \rowRot{ \makecell[l]{Choosable inputs} } &
        \rowRot{ \makecell[l]{Communication\\Transfer Rounds} } &

        \rowRot{ \makecell[l]{Resilience to\\malicious adv.} } &
        \rowRot{ \makecell[l]{No Trusted Third \\ Party required} } &
        \rowRot{ \makecell[l]{Cryptographic\phantom{x}\\primitives$^\P$} }

        \rowRot{ \makecell[l]{Proof-of-concept \\ implementation$^\ddagger$} } &
        \rowRot{ \makecell[l]{Transfer global\\exec.\ time [\textmu s]} } &
        \rowRot{ \makecell[l]{Sender execution\\time [\textmu s]} } &
        \rowRot{ \makecell[l]{Receiver\\online exec.\ [\textmu s]} } &
        \rowRot{ \makecell[l]{Receiver (IoT)\\online exec.\ [\textmu s]} } &

        \rowRot{ \makecell[l]{Allows for Receiver\\pre-computation$^\star$} } &
        \rowRot{ Open Source } &
        \rowRot{ Language } \\

        \midrule

        1985 & \cite{Even1985} & \chkmark & 3 & \crossmark & \chkmark & PKE & \crossmark & \implnotice \\
        1989 & \cite{bellare1989non} & \chkmark & 1 & \crossmark & \chkmark & CDH & \crossmark & \implnotice \\
        1999 & \cite{rivest1999unconditionally} & \chkmark & 2 & \chkmark & \crossmark & \makecell[c]{Trusted initializer} & \crossmark & \implnotice \\
        2004 & \cite{tzeng2004efficient} & \chkmark & 2 & \chkmark & \chkmark & DDH & \crossmark & \implnotice \\
        2005 & \cite{crepeau2005efficient} & \crossmark & 2 & \chkmark & \crossmark & \makecell[c]{Noisy channel} & \crossmark & \implnotice \\
        2008 & \cite{peikert2008framework} & \chkmark & 2 & \chkmark & \chkmark & PKE (DDH/LWE) & \crossmark & \implnotice \\
        2011 & \cite{ishai2011constant} & \chkmark & 2 & \chkmark & \chkmark & ROM & \crossmark & \implnotice \\
        2012 & \cite{lindell2012secure} & \chkmark & 2 & \chkmark & \chkmark & DDH & \crossmark & \implnotice \\

        2015 & \cite{asharov2015malicious} & \chkmark & 2 & \chkmark & \chkmark & DDH & \chkmark & $\sim 3680$ & \notprovidednotice & \crossmark & \chkmark & C (of \cite{peikert2008framework}) \\
        2015 & \cite{CO-OT15} & \crossmark & 2 & \chkmark & \chkmark & DH & \chkmark & 114 & \notprovidednotice & \chkmark & \chkmark & Assembly \\
        $\S$ & \cite{CO-OT15} & -- & -- & -- & -- & -- & -- & 257 & 76.76 & 68.91 & 567.04 & \chkmark & \chkmark & C \\

        2019 & \cite{liu2019universally} & \chkmark & 2 & \chkmark & \chkmark & \makecell[c]{Ideal lattice} & \crossmark & \implnotice \\
        2019 & \cite{masny2019endemic} & \chkmark & 2 & \chkmark & \chkmark & PKE (LWE, RO) & \chkmark & $\sim 200$ & \notprovidednotice & \chkmark & \chkmark & Assembly \\

        2020 & \cite{dottling2020two} & \chkmark & 2 & \chkmark & \chkmark & CDH & \crossmark & \implnotice \\
        2020 & \cite{aragon2020croot} & \chkmark & 2 & \chkmark & \chkmark & ROM (RQC-III) & \chkmark & 1760$^\dagger$ & 503.33$^\dagger$ & 1080$^\dagger$ & \notprovidednoticeshort & \chkmark & \crossmark & unspecified \\
        2020 & \cite{10.1007/978-3-030-45727-3_23} & \chkmark & 3 & \chkmark & \chkmark & LWE & \crossmark & \implnotice \\
        2020 & \cite{canetti2020blazing} & \crossmark & 3 & \chkmark & \chkmark & CDH & \chkmark & 164 & \notprovidednotice & (\chkmark) & \crossmark & Relic \\

        2021 & \cite{lai21compact} & \crossmark & 2 & \chkmark & \crossmark & Isogeny & \crossmark & \implnotice \\

        2022 & \cite{badrinarayanan2022efficient} & \chkmark & 2 & \chkmark & \chkmark & \makecell[c]{PKE (ROM)} & \crossmark & \implnotice \\

        2024 & \cite{branco2024two} & \chkmark & 2 & \chkmark & \chkmark & QR + LPN & \crossmark & \implnotice \\

        2024 & \cite{abadi2024supersonic} & \chkmark & 4 & \crossmark & \crossmark & NI & \chkmark & $\sim$350 & \notprovidednotice & \crossmark & \chkmark & C++ \\
        2024 & \cite{quietOT} & \crossmark & 2 & \chkmark & \chkmark & NI + LWE + ROM & \chkmark & N/A$^{\Delta}$ & \notprovidednotice & \chkmark & \chkmark & C \\

        2025 & \cite{berti2025lr} & \chkmark & 2 & \crossmark & \chkmark & DDH & \crossmark & \implnotice \\

        2025 & \cite{yang2025quantum} & \chkmark & 2 & \chkmark & \crossmark & \makecell[c]{Quantum ent.} & \crossmark & \implnotice \\

        2025 & \cite{yang2025all} & \crossmark & 2 & \chkmark & \chkmark & \makecell[c]{Quantum ent.} & \crossmark & \implnotice \\
        2026 & \cite{abadi2026oblivis} & \chkmark & 4 & \crossmark & \crossmark & NI & \chkmark & 530 & \notprovidednotice & \notprovidednoticeshort & \chkmark & C++ \\

        2026 & \bf{\IOT} & \chkmark & 2 & \chkmark & \chkmark & QR & \chkmark & 2366 & 2350.93 & 2.80 & 39.90 & \chkmark & \chkmark & C/Python \\

        \bottomrule

        \end{tabular}
    }}
    \\[0.5em]
    {\footnotesize\noindent
        $^{\P}$Acronyms:
            (ROM) Random Oracle Model;
            (CDH) Computational Diffie--Hellman;
            (DDH) Decisional Diffie--Hellman;
            (NI) Non-interactive tech.;
            (LWE) Learning With Errors;
            (QR) Quadratic Residues;
            (LPN) Learning Parity with Noise.
        \\[0.1em]
        \noindent$^\ddagger$ Computational security is set to 128 bits in all works, but \cite{abadi2024supersonic}, which relies on unconditional security based on NI.
        \\[0.1em]
        \noindent$^\star$ (\chkmark) allows pre-computation, yet with minor impact on overall performance ($<2\times$).
        \\[0.1em]
        \noindent$^{\Delta}$ OT Extension relies on non-interactive PKE setup, hence no baseOT is assessable.
        \\[0.1em]
        \noindent$^{\dagger}$ Performances evaluated in CPU cycles and estimated at 3GHz.
        \\
        \noindent$^\S$ Results obtained over our porting of the library.
    }
\end{table*}

Oblivious Transfer (OT), first introduced by Rabin in 1981~\cite{Rabin1981}, laid the foundation for secure multi-party computation by proposing a method to exchange secrets without either party learning superfluous information. Rabin's protocol leverages quadratic residues and public-key cryptography to ensure the Sender remains oblivious to the recipient's choice, inspiring extensive research on OT efficiency and security.

Building upon Rabin's foundation, Even, Goldreich, and Lempel formalised the one-out-of-two variant of OT in 1985~\cite{Even1985}, establishing it as a cornerstone for secure multi-party computation, private information retrieval, and privacy-preserving cryptographic protocols. We refer the reader to Yadav et al.~\cite{Yadav2022} for a broader survey of OT protocols and their applications.

\OT protocols can be extended to OT$_1^n$ and OT$_k^n$, allowing the Receiver to select one or multiple messages out of $n$ while preserving obliviousness and message privacy~\cite{naor1999oblivious}. Such extensions underpin advanced cryptographic functionalities, including secure function evaluation and private set intersection. As a result, OT has been extensively studied under a variety of assumptions, including classical number-theoretic assumptions (\eg, DDH, CDH), lattice-based primitives, and generic dual-mode public-key encryption frameworks~\cite{peikert2008framework}.

Naor and Pinkas~\cite{naor2001efficient} improved OT by constructing efficient 1-out-of-$n$ OT schemes from the DDH assumption, with logarithmic overhead in the selection domain and round-efficient communication, becoming a reference point for many practical implementations. Subsequent works refined OT under different security and communication models, including adaptive queries~\cite{naor1999oblivious}, noisy channels~\cite{crepeau2005efficient}, and information-theoretic Receiver privacy~\cite{tzeng2004efficient}.

A closely related primitive is Random Oblivious Transfer (ROT), also referred to as OT with random inputs, in which the Sender does not choose the messages but instead learns two uniformly random strings generated by the protocol, while the Receiver learns exactly one of them. Introduced and formalised by Cr\'epeau in the context of noisy-channel cryptography~\cite{crepeau2005efficient}, ROT was shown to be equivalent to standard \OT, as each primitive can be obtained from the other via local post-processing. This equivalence has made ROT the preferred abstraction in modern protocol design, particularly in settings where correlated randomness is generated independently of application-level inputs.

A major line of work focuses on OT extensions. Ishai et al.~\cite{ishai2003extending} introduced the OT extension paradigm, enabling a large number of OTs using only a small number of expensive base (R)OTs and predominantly symmetric-key operations. This paradigm was later refined to achieve stronger guarantees, including malicious security~\cite{asharov2015malicious}, and is now the de facto standard for large-scale secure computation. More recently, Couteau et al.~\cite{quietOT} proposed QuietOT, an OT-extension framework with a public-key setup, allowing parties to generate pseudorandom OT correlations after a one-time public-key phase and with only lightweight online operations. Unlike traditional OT extension, QuietOT reduces the need for an interactive base-OT setup. As a result, however, the concrete cost of the initial batch of base OTs is often treated as negligible, an assumption that holds only when the number of extended OTs is sufficiently large.

Beyond OT extension, several frameworks construct OT generically from public-key encryption (PKE). Peikert, Vaikuntanathan, and Waters~\cite{peikert2008framework} proposed a composable framework yielding UC-secure OT in the CRS model under assumptions such as DDH, quadratic residuosity, and LWE. Follow-up works, including endemic OT~\cite{masny2019endemic} and subsequent optimizations~\cite{badrinarayanan2022efficient}, further improved efficiency and multi-user security. However, public-key operations remain computationally expensive and often impractical on low-power devices.
A complementary line of inquiry concerns the formal relationship between OT and Private Information Retrieval (PIR), in which a Receiver privately fetches an entry from a remote database. The two primitives are tightly intertwined: Di Crescenzo et al.~\cite{diCrescenzo2000pir} showed that any non-trivial single-server PIR implies OT, while conversely 1-out-of-$n$ OT directly realises symmetric PIR (SPIR). Modern PIR research has substantially advanced the communication efficiency of single-server retrieval, yet these constructions typically rely on heavy public-key or homomorphic primitives that remain prohibitive on constrained Receivers --- a constraint that further motivates lightweight base OT designs of the kind proposed here.

Considerable effort has also been devoted to optimising the concrete efficiency of base OT protocols. The Simplest OT protocol by Chou and Orlandi~\cite{CO-OT15} (whose security has been revised in~\cite{CO-OT18} and further re-examined in~\cite{postCO-OT15}) achieves extremely low concrete cost using elliptic-curve cryptography (ECC) and is widely adopted in practice, despite offering weaker security guarantees than modern simulation-based frameworks. More recent protocols, such as Blazing OT~\cite{canetti2020blazing}, achieve UC security with competitive performance, but still rely on ECC and often benefit from hardware acceleration, which is typically unavailable on constrained low-power platforms.

Several works further refine communication and round complexity in the batch setting; for example, Branco, D\"ottling, and Srinivasan~\cite{branco2024two} propose a two-round batch OT protocol with near-optimal communication based on quadratic residuosity and the Learning Parity with Noise assumption. Alternative directions explore different assumptions and frontiers, including LWE-based OT~\cite{10.1007/978-3-030-45727-3_23}, isogeny-based OT~\cite{lai21compact}, and quantum or post-quantum OT variants~\cite{sarkar2024efficient,zhang2023practical}, typically targeting powerful platforms.

Other applications of the OT protocol include machine learning and AI implementation. These studies~\cite{gan2025cuot,lin2025ironman} demonstrate significant throughput improvements in large-scale secure computation and privacy-preserving machine learning by exploiting massive parallelism and hardware acceleration.

In parallel to classical software-oriented OT designs, several recent works also explore orthogonal directions that differentiate from the traditional setting. A first line of research investigates quantum variants of OT, extending the primitive to quantum messages or even to unknown unitary operations applied remotely: these protocols rely on entanglement~\cite{yang2025quantum}, Bell-state measurements~\cite{yang2025all}, and quantum communication to achieve all-or-nothing semantics, and are primarily motivated by the intrinsic privacy-preserving properties inherently provided by distributed quantum computations.

Finally, more closely related to the present contribution, Abadi and Desmedt~\cite{abadi2024supersonic} introduced Supersonic OT, an ultra-efficient OT extension protocol achieving unconditional security without public-key cryptography, later generalized to OT$_1^n$ and OT$_k^n$~\cite{Abadi25}.
A step forward is provided in~\cite{abadi2026oblivis}, where SupersonicOT is embedded in a more general framework offering delegated-query and multi-receiver OT, while targeting also low-end devices.
Although achieving remarkable concrete performance, the solution relies on a trusted preprocessed setup, a trusted proxy mediation, and is not devised to decouple sender and receiver interactions.

The connection between OT and quadratic residues originates from Rabin's protocol~\cite{Rabin1981} and is further developed in modern QR-based constructions~\cite{peikert2008framework,branco2024two}. Rather than aiming for maximum generality or asymptotic efficiency, our construction explicitly targets scenarios where the Receiver is severely resource-limited. By delegating most cryptographic workload to the Sender or to an offline phase, the Receiver's online computation is reduced to a minimal set of operations, making the protocol practical even on IoT devices. Table~\ref{tab:comparison_table} summarises representative OT protocols and contrasts them with our proposal~\IOT in terms of structural properties and computational analysis.


\section{Notation and Background}
\label{sec:notation}

\begin{table}[t]
    \centering

    \def\arraystretch{1.2}
    \caption{
        Table of symbols, divided by phase and ownership.
    }
    \label{tab:notations}
    \setlength{\tabcolsep}{8pt}
    \footnotesize
    \begin{tabular}{c|c|c|l}
        \toprule
        \textbf{Var} & \textbf{Definition} & \textbf{Owner} & \textbf{Description} \\
        \midrule
        \midrule
        \multicolumn{3}{c}{\textbf{Setup}}\\
        \midrule
        $p, q$ & $\equiv_41$ & Sender & RSA-like prime numbers\\
        $n$ & $=p\cdot q$ & Shared & Integer modulus\\
        $\ZZ_n$ & --- & Shared & Ring of integers modulo $n$ \\
        $\lambda$ & $=\log_2(n)$ & --- & Bit size of the modulus $n$ \\
        $\kappa$ & --- & --- & Security parameter \\
        $F$ & $\KKK\times Z_n \to \YYY$ & --- & Family of KDF \\
        $\KKK$ & --- & --- & Key space for the KDF\\
        \multirow{1}{*}{$\YYY$} & \multirow{1}{*}{---} & \multirow{1}{*}{---} & Codomain of the KDF (message space)\\
        \multirow{1}{*}{$H$} & \multirow{1}{*}{$\ZZ_n \to \HHH$} & \multirow{1}{*}{---} & Cryptographic hash function (used for roots digests)\\
        $\HHH$ & --- & --- & Root tag/digest space\\
        $I$ & $=\sqrt{-1}$ & Sender & Square root of $-1$ mod $n$\\
        \midrule
        \multicolumn{3}{c}{\textbf{Transfer phase}}\\
        \midrule
        $M_0, M_1$ & $\in \YYY$ & Sender & Oblivious transfer messages \\
        $k$ & $\in \ZZ_n$ & Receiver & Oblivious transfer session key \\
        $d$ & $H(k)$ & Receiver & Session key digest \\
        $t$ & $=k^2 \bmod n$ & Receiver & Residue of the session key \\
        $\hat k$ & $\in\ZZ_n$ & --- & Conjugate of the session key \\
        $b$ & $\in \{0, 1\}$ & Receiver & Oblivious transfer choice\\
        $r$ & $\in\ZZ_n$ & Shared & Residue hiding the choice\\
        $s$ & $\in\KKK$ & Shared & Key for the KDF $F$\\
        $\{k_i^{(j)}\}$ & $\in\ZZ_n$ & Sender & Set of plausible session keys\\
        $i$ & $\in \{b, 1-b\}$ & Shared & Receiver possible choices\\
        $j$ & $\in \{0, 1\}$ & Shared & Receiver conjugate index\\
        $\hat j$ & $\in \{0, 1\}$ & Receiver & Actual conjugate index\\
        $\{c_i^{(j)}\}$ & $=M_i\oplus K_i^{(j)}$ & Shared & Set of ciphertexts\\
        $\{d_i^{(j)}\}$ & $=H(k_i^{(j)})$ & Shared & Plausible session key digests\\
        $\{K_i^{(j)}\}$ & $=F_s(k_i^{(j)})$ & Sender & Set of plausible message keys\\
        $K$ & $=F_s(k)$ & Receiver & Actual message key\\
        \midrule
        \multicolumn{3}{c}{\textbf{Check congruency phase}}\\
        \midrule
        $\ell$ & --- & Shared & Number of tests\\
        $X$ & $\in \ZZ_n^\ell$ & Receiver & Test set\\
        $c$ & $\in \{1, 2\}$ & Receiver & Challenge choice\\
        $y$ & $=-x^c \bmod n$ & Shared & Challenge value\\
        $r$ & $\in\{\texttt{T}, \texttt{F}\}$ & Shared & Challenge response\\
        $R$ & $\{r \mid c = 1\}$ & Receiver & Response set of $c=1$ answers\\
        \bottomrule
    \end{tabular}
\end{table}

In what follows, we identify with $\ZZ_n$ the ring of integers modulo $n = p \cdot q$ (with $p, q$ primes) and we say that a given $a \in \ZZ_n$ is positive if $a < \sfrac n2$. All of the operations, if not differently specified, are carried out in $\ZZ_n$.

We state with $\lambda$ the size of the modulus $n$ and we report with $\kappa$ the corresponding (symmetric) security parameter (\eg $\lambda = 3072 \Rightarrow \kappa = 128$, see \cite{NIST-security}).

Quadratic residues modulo $n$ are identified by the set $$\QR_n := \{x \in \ZZ_n \mid \exists t \in \ZZ_n,\ t^2 \equiv_n x\}$$ (or simply $\QR$) and we recall that evaluating whether $x\in\QR$ is considered hard, the so-called \textit{Quadratic Residuosity Assumption} (QRA), unless the factorisation of $n$ is known; in such case, many different approaches do exist, including the \textit{Tonelli-Shanks} algorithm which also allows evaluating the two couples of roots in $\ZZ_p$ and $\ZZ_q$ which can be combined into four roots in $\ZZ_n$ according to the \textit{Chinese Remainder Theorem} (CRT), \ie solving a modular system of two equations.

We recall a \emph{Key Derivation Function} (KDF) being a function transforming keying material into one or more cryptographically strong outputs.
Typical goals are entropy stretching and domain separation, with Kessak-based KDF (SHAKE-based, \cite{SHA-3}) being a widely deployed construction.
Here, different useful properties do hold, including pre-image and second-preimage resistance.
More formally, the core of KDFs can be abstracted as \emph{Pseudorandom Functions} (PRFs), \ie keyed family of functions
\[
    F: \KKK \times \XXX \to \YYY
\]
such that, for any random key $k \in \KKK$, $F_k(\cdot)$ is computationally indistinguishable from a random function.
In what follows, we assume $\XXX=\ZZ_n$.

We define a message $M$ as a generic bit-string, $M \in \{0, 1\}^* = \PPP$, and we denote with $\oplus$ the bit-wise xor.
In what follows, we assume $\PPP = \YYY$, namely the message is compatible with the output of $F$ so that messages can be efficiently encrypted \textit{\`a la} One Time Pad (OTP) as $c = M \oplus y$.

Finally, we define with $H$ a generic cryptographic hash function over $\ZZ_n$:
$$H : \ZZ_n \to \HHH\,.$$
In our construction, $H(\cdot)$ is later used to discriminate over quadratic residues (see later $d_i^{(j)}$), hence it is important to recall that pre-image and second-preimage resistance hold;
however, digests of $\HHH$ might be modest in size, if compared with $n$.

For the reader's convenience, we provide in Table~\ref{tab:notations} a list of all symbols adopted within the manuscript, arranged according to the protocol's phases.


\section{The \IOT Protocol}
\label{sec:IOT2}

In this section, we describe the architecture of the proposed \IOT model, our novel lightweight OT encryption scheme based on the problem of quadratic residuals, able to deliver an \OT-transfer to the Receiver requiring only a few transmissions and operations.
We recall that operations are carried out within the ring $\ZZ_n$ if not differently specified.
Here $n = p\cdot q$, with $p, q$ being primes chosen by the Sender such that $p, q \equiv_4 1$ (this ensures that $-1 \in \QR$, \ie $\exists I \in \ZZ_n$ such that $I^2 \equiv_n 1, I \ne 1$, $I := \sqrt{-1}$) and (eventually) validated by the Receiver with a dedicated protocol (see later Section~\ref{ssec:untrustedServer} for a discussion on untrusted senders).


\subsection{The system and security models}
\label{ssec:system-model}

We rely on the classical OT model, which we briefly recall here for completeness.
The system consists of two logical entities interacting over a public communication channel
\begin{description}
	\item[Sender] The Sender holds pairs of messages $(M_0, M_1)$ and is responsible for executing the transmission protocol.
    The Sender is assumed to be computationally capable, and can efficiently perform expensive operations such as square root extraction in $\mathbb{Z}_p \times \mathbb{Z}_q$.
    For what concerns security properties, the server might be modelled as an unbounded party (being the security of the choice obliviousness, see below, proved under unconditional, information-theoretic, guarantees).
	\item[(Set of) Receiver(s)] The Receiver selects a bit $b \in \{0,1\}$ and interacts with the Sender to retrieve specifically the message $M_b$.
    Receivers are assumed to be computationally constrained and are required to perform only lightweight operations.
    Importantly, our security guarantees do not rely on this assumption, and remain valid even when Receivers are capable of heavy-computation (or can delegate them to others).
\end{description}
We can then formalise the scheme as follows.
\begin{cryptoscheme}[\OT]\label{scheme:ot}
\hfill\\
    An \OT~scheme consists of the following three algorithms:
	\begin{description}
		\item [$\setup(\kappa)\to(\sk, \pk)$:] is the initialisation algorithm run by the Sender to set up a pair of public keys ($\pk = n$) and a secret key ($\sk = (p, q)$) starting from a security parameter $\kappa$ ($\kappa$ determines the length $\lambda$ of the generated modulus).
		\item [$\transfer(\pk, b)\to (M_b, \bot)$:] is the online interactive transfer that is run by the Receiver who selects a message $b \in \{0, 1\}$ and queries the Sender for the two messages $M_0$ and $M_1$ (properly encrypted), of which only $M_b$ will be decryptable by the Receiver itself.
		\item [$\checkreveal()\to\sk$:] is an optional phase where the Receiver queries the Sender for its secret key so that he can check the entire protocol was played fairly.
	\end{description}
    In particular, the transfer phase is required to satisfy the following security properties:
    \begin{description}
        \item[Choice obliviousness] The Sender and any external attacker learn no information regarding the Receiver's choice bit $b$.
        \item[Messages obliviousness] The Receiver learns no information about $M_{1-b}$, while external adversaries learn nothing about either message.
    \end{description}
\end{cryptoscheme}

Do notice that the optional \checkreveal phase allows the Receiver to verify that the protocol was executed fairly.
In fact, while the Receiver is unable to cheat the Sender by design, the Sender may trivially send identical (or slightly modified) messages as $(M_0, M_1)$, a typical drawback affecting OT-protocols.
The \checkreveal mechanism hence enables post-hoc detection of such misbehaviour, making it a valuable tool for the Receiver.

For what concerns protocol security, it is in general not possible to provide unconditional (information-theoretic) guarantees on both the Receiver and the Sender at the same time.
In particular, as we will see later in Section~\ref{sec:security-analysis}, our design offers unconditional choice obliviousness while granting computational message security.


\subsection{The threat models}
\label{ssec:threat-models}

We are now ready to describe the threat model for \IOT.
Let us consider the following four entities involved in the protocol:
\begin{description}
    \item[Sender (s-)] In what follows, we assume the Sender being semi-trusted (honest-but-curious), meaning that it follows the protocol honestly (\ie, it does not deviate from the prescribed steps) but tries to learn as much as possible about the choice $b$ from the communications it receives during the execution of the protocol itself. We will discuss specifically the security of malicious senders in Section~\ref{ssec:untrustedServer}.
    \item[Receiver(s) (c-)] We assume the Receiver(s) as untrusted, meaning that it (they) can deviate from the protocol and use any strategy to try breaking the messages' obliviousness or the protocol's cryptographic security.
    \item[Passive attacker (p-)] We consider computational bounded passive attacker capable of eavesdropping all the communication occurring on the public channels, including multiple independent communications.
    \item[Active attacker (a-)] We suppose active attackers being able to control the communication channel and, hence, capable of read, intercept, modify, and generate messages at will. We do not consider poisoning attacks since we discuss the security against malicious receivers in what follows and against malicious senders later.
\end{description}

In particular, we identify the following eight threats:
\begin{description}
    \item[(s-i)] Sender breaking choice obliviousness: the Sender successfully retrieves the choice $b$ operated by the Receiver during the transfer phase.
    \item[(c-i)] Receiver breaking cryptographic primitive: the Receiver is able to manipulate the protocol to break the security of the key $n$, \ie retrieve its factorisation $p, q$.
    \item[(c-ii)] Receiver breaking messages obliviousness: the Receiver is not able to break the primitive, but it is able to manipulate the protocol to successfully retrieve both messages $M_0$ and $M_1$ regardless of the choice $b$ operated.
    \item[(c-iii)] Receivers collusion: multiple receivers adopting the very same Sender key-pair can collude to gain a non-negligible advantage on any two of the previous tasks.
    \item[(p-i)] Eavesdropping attack: An external attacker can retrieve either the message(s) or the Receiver's choice by solely monitoring the traffic.
    \item[(a-i)] Receiver impersonation attack: The attacker can impersonate the Receiver during the transfer phase to retrieve one (or both) the messages.
    \item[(a-ii)] Sender impersonation attack: The attacker can impersonate the Sender during the transfer phase to communicate a fake, valid message (for breaking choice obliviousness, see \textbf{(s-i)}).
    \item[(a-iii)] MitM attack: An external attacker can break into an active transfer to retrieve the messages, the bit choice, or to transmit chosen messages by acting as a man-in-the-middle, \ie potentially modifying the messages passing through the public channel without the parties noticing it.
\end{description}

\begin{table}[!t]
    \centering
    \caption{Threat model and security assumptions.}
    \label{tab:threat_models}
    \footnotesize
    \setlength{\tabcolsep}{8pt}
    \def\arraystretch{1.2}
    \begin{tabular}{c|c|c|c}
        \toprule
        Attacker & Model & Target & Security assumption \\\midrule\midrule
        Sender & \textbf{(s-i)} & Choice obliviousness & Unconditionally secure\\
        \midrule
        \multirow{3}{*}[-5pt]{Receiver} & \textbf{(c-i)} & Factorisation of $n$ & \multirow{2}{*}[-3pt]{\makecell{indistinguishability \\ + reduction}}\\
        \cmidrule{2-3}
        & \textbf{(c-ii)} & Message obliviousness & \\
        \cmidrule{2-4}
        & \textbf{(c-iii)} & Collusion & tNM-CCA\\
        \midrule
        Passive & \textbf{(p-i)} & Eavesdrop messages & \makecell{indistinguishability \\ + tNM-CCA} \\
        \midrule
        \multirow{3}{*}[-5pt]{Active} & \textbf{(a-i)} & Receiver impersonation & --$^\dagger$\\
        \cmidrule{2-4}
        & \textbf{(a-ii)} & Sender impersonation & \multirow{2}{*}[-3pt]{\makecell{Public-key \\ security}}\\
        \cmidrule{2-3}
        & \textbf{(a-iii)} & MitM & \\
        \bottomrule
    \end{tabular}\\[.5em]

    $^\dagger$Can be achieved with Receiver single-side authentication layer.
\end{table}

We briefly discuss the corresponding security in what follows, and we resume the key concepts in Table~\ref{tab:threat_models}.

The choice obliviousness \textbf{(s-i)} is unconditionally secure by design, while the security of the cryptographic primitive \textbf{(c-i)} and the messages obliviousness \textbf{(c-ii)} for single messages is guaranteed under the properties of the KDF and secure hash functions combined with a classical reduction to the factorisation problem. Receivers' collusion-resistance \textbf{(c-iii)}, as well as security for \textbf{(c-i)} and \textbf{(c-ii)} against message combination, are discussed under chosen ciphertext attacks with tag-non-malleability (see also~\cite{McQuoid2021}). In particular, here we recall a protocol being non malleable (NM) whether an attacker is not able to create a valid ciphertext with some arbitrary correlation with known ciphertexts~\cite{Bellare1998}.
Tag-non-malleability under Chosen ciphertext Attack (tNM-CCA1) implies in this context that the Receiver is not able to reconstruct a valid ciphertext with tag $s$ without breaking the security of the protocol, even if it is able to decrypt polynomially many OT instances using different tags $s' \ne s$~\cite{Pass2005}. We also recall that NM-CCA1 also implies indistinguishability under the same assumption (IND-CCA1, \cite[3.1]{Bellare1998}).

Eavesdropping resistance \textbf{(p-i)} is obtained under ciphertext indistinguishability and tag-non-mal\-le\-a\-bi\-li\-ty.

Conversely, Sender impersonation during the transfer phase \textbf{(a-ii)} is unfeasible due to the impossibility of breaking the public key (as it is either for the Receiver itself).
As a consequence, MitM attacks \textbf{(a-iii)} are unfeasible as well, as the attacker would be unable to impersonate the Sender with the Receiver or retrieve the messages encrypted with the Receiver key.
Finally, being the Receiver not involved during the setup phase, an attacker is always able to impersonate a not-authenticated Receiver \textbf{(a-i)}; however, this is endemic with the setting itself and can be easily mitigated by adding a single-side authentication layer between the Receiver and the Sender during the transfer phase.


\subsection{The scheme}
\label{ssec:scheme}

In what follows, we present the three algorithms, also depicted in Figures~\ref{fig:setup} and~\ref{fig:transfer}, that build our method as presented in Scheme~\ref{scheme:ot}.

\subsubsection{Setup phase}

During the setup phase, the Sender is required to build a public modulus $n$ as a public key-pair of size $\lambda$, with $\lambda$ chosen according to~\cite{NIST-security} depending on the required security level $\kappa$.
In particular, it is important to notice that the setup requires no interaction with the Receiver, actually making it a Sender-exclusive phase: this also allows the Sender to use the same setup even with different receivers, which can potentially help each other to ensure the trustworthiness of the Sender (see later Section~\ref{ssec:untrustedServer}).

In detail (see also Figure~\ref{fig:setup}), the Sender chooses two large primes $p, q$ exercising all due caution, like requiring that $p - q$ is big, and evaluates $n = p\cdot q$. In particular, values are chosen such that $n \sim 2^\lambda$ and $p, q \equiv_4 1$ must hold as this ensures that $I \in \QR$ mod $n$, which is required for the protocol to be oblivious from the Sender side.
In this phase, the Sender can also pre-evaluate $I$ to later speed up the transfer phase.

\begin{figure}[!t]
    \centering
    \def\arraystretch{1.5}\small
    \begin{tabular}{|L{5.5cm}C{2.0cm}L{5.5cm}|}
        \hline
        \textbf{Sender} & \textbf{Public} & \textbf{Receiver}\\
        \hline
        \textbf{choose security} \textbf{parameter} $\kappa, \lambda$ && \\
        \textbf{choose} $p \ne q$ \textbf{primes} && \\
        \phantom{xx}\textbf{s.t.}\ $p\cdot q \sim 2^\lambda$ \textbf{and}\ \ $p, q \equiv_4 1$ &&\\
        $n \lassign p\cdot q$ & $\xmapsto[\phantom{xxxxx}]{n}$ & $n$\\
        \textbf{precompute} $I = \sqrt{-1}$ & & \\
        \hline
    \end{tabular}
    \caption{Setup phase of the oblivious transfer.}
    \label{fig:setup}
\end{figure}

\subsubsection{Transfer phase}

Let $M_0, M_1 \in \YYY$ be the two messages the Sender is willing to transfer to the Receiver.

Conversely to a classical \OT where the transfer phase is opened with the Sender sending two nonces to the Receiver to be used as message keys/discriminant, in our protocol the Receiver session key $k$ is unique and it is generated by the Receiver itself as a random positive value in $\ZZ_n$ (the usual caution of choosing $k>\sqrt{n}$ and $\gcd(k, n) = 1$ must be exerted).
This is particularly relevant since, apart from reducing the number of required rounds, it allows the Receiver to pre-evaluate $k$ and subsequently pre-compute its corresponding residue $t=k^2 \pmod n$ and its digest $d = H(k)$.

The Receiver then operates a choice on which message to pick, say $b \in \{0, 1\}$, and masks $k$ by squaring it and communicating either $r = t$ or $r = -t \pmod n$ (\ie $r = n-t$), depending on whether $b=0$ or $b = 1$ respectively.

Here it is important to recall that under the assumption $p, q \equiv_4 1$, $x \in \QR$ has exactly four roots in the form $a, -a, b, -b$ and, since both $-1, k^2 \in QR$, we have that both $r, -r \in \QR$. In fact, if $b = 0$ we have
\[
    \begin{cases}
        r = t = k^2 & \pmod n\\
        -r = -t = -k^2 & \pmod n
    \end{cases}\,,
\]
while if $b = 1$ we have
\[
    \begin{cases}
        r = -t = -k^2 & \pmod n \\
        -r = t = k^2 & \pmod n
    \end{cases}
    \ .
\]
It follows that the Sender is able to compute both $\sqrt{r}$ (mod $n$) and $\sqrt{-r}$ (mod $n$) using the pre-computed $I$;
let $k_0^{(0)}, k_0^{(1)}$ be the two positive roots of $r$ and let $k_1^{(0)}, k_1^{(1)}$ be the two positive roots of $-r$.

In particular, do notice that $k \in \{k_i^{(j)}\}_{i, j \in \{0, 1\}}$. However, the Sender is not able to determine $\hat i \equiv b$, that is $\hat i \in \{0, 1\}$, such that $k \in \{k_{\hat i}^{(j)}\}_{j \in\{0, 1\}}$.
Concurrently, the Receiver must remain oblivious to the second root $\hat k$ in $\{k_{\hat i}^{(j)}\}$ as it is conjugate with $k$, namely $\{k_{\hat i}^{(j)}\} = \{k, \hat k\}$ and its disclosure would break the factorisation of $n$.
In this regard, let $s$ be a transfer nonce sampled by the Sender from $\KKK$ and communicated to the Receiver.
Here $s$ is used to agree with the Receiver upon a KDF $F_s(\cdot)$ to hide the values of $\{k_i^{(j)}\}$, hence generating four plausible session keys $K_i^{(j)} = F_s(k_i^{(j)})$.

Being randomly distributed due to the properties of $F_s(\cdot)$, the Sender can then use the two couples of session keys to encrypt the messages $M_0$ and $M_1$ with a simple bit-wise xor operation, hence producing four ciphertexts of the form $c_i^{(j)} = M_i \oplus K_i^{(j)} = M_i \oplus F_s(k_i^{(j)})$.
These ciphertexts are finally sent to the Receiver alongside with the digest of the four plausible roots $d_i^{(j)} = H(k_i^{(j)})$ for the Receiver to be able to match the correct $\hat i = b$ and $\hat j$\footnote{\label{fn:replaceHashWithMask}Do note that here a possible alternative is to use some sort of small tag $f$ to be prefixed/postfixed to the message instead of sending $\{d_i^{(j)}\}$. In particular, this avoids executing the hash operation at the cost of restricting the message space by the size (in bits) of $f$. Left for future work.}.

Upon delivery, the Receiver is able to determine which message $c_b^{(\hat j)}$ he can decrypt with $k$ by checking which digest $d_b^{(j)}$ equals the precomputed $d=H(k)$.
The transfer phase is finally concluded by decrypting $M_b$: this can be done by simply evaluating $K = F_s(k)$ and xor-ing it from the correct ciphertext $c_b^{(\hat j)}$.

\begin{figure*}[!t]
    \centering
    \def\arraystretch{1.5}
    \scalebox{.97}{\begin{tabular}{|L{5.5cm}C{2.0cm}L{5.5cm}|}
        \hline
        \multicolumn{3}{|c|}{Pre-computation (Offline)}\\
        \hline
        \textbf{Sender} & & \textbf{Receiver}\\
        \hline
        \textbf{retrieve:} & & \textbf{retrieve:}\\
        \phantom{xx} $n$, $p$, $q$, $I$ & & \phantom{xx} $n$\\
        \textbf{precompute:} & & \textbf{precompute:}\\
        \phantom{xx} $s \sample \ZZ_n$ && \phantom{xx} $k \sample \ZZ_{\sfrac n2}$ \textit{s.t.}\ $\gcd(k, n) = 1$\\
        && \phantom{xx} $t \lassign k^2 \quad\mod n$\\
        && \phantom{xx} $d \lassign H(k)$\\
        \hline\multicolumn{3}{c}{}\\\hline
        \multicolumn{3}{|c|}{Transfer (Online)}\\
        \hline
        \textbf{Sender} & \textbf{Public} & \textbf{Receiver}\\
        \hline
        $M_0, M_1 \in \YYY$
        &&\textbf{choose} $b \in \{0, 1\}$ \textbf{message}\\
        $r$ & $\xmapsfrom[\phantom{xxxxxxx}]{r}$ & $r \lassign (-1)^b \cdot t \quad\mod n$\\
        & $\xmapsto[\phantom{xxxxxxx}]{s}$ & $s$\\
        $\begin{rcases*} k_0^{(0)} \\ k_0^{(1)} \end{rcases*} \lassign \sqrt{r}$ \;\;\;\;\; \textit{, chosen in $\ZZ_{\sfrac n2}$}&& $K \lassign F_s(k)$\\[1em]
        $\begin{rcases*} k_1^{(0)} \\ k_1^{(1)} \end{rcases*} \lassign \sqrt{r}\cdot I$ \; \textit{, chosen in $\ZZ_{\sfrac n2}$}&&\\
        $c_i^{(j)} \lassign M_i \oplus F_s\big(k_i^{(j)}\big), \quad i, j = 0, 1$ & $\xmapsto[\phantom{xxxxxxx}]{\{c_i^{(j)}\}}$ & $c_0^{(0)}, c_0^{(1)}, c_1^{(0)}, c_1^{(1)}$\\
        $d_i^{(j)} \lassign H(k_i^{(j)}), \quad i,j = 0, 1$ & $\xmapsto[\phantom{xxxxxxx}]{\{d_i^{(j)}\}}$ & $d_0^{(0)}, d_0^{(1)}, d_1^{(0)}, d_1^{(1)}$\\
        && \textbf{let} $\hat j \in \{0, 1\}$ \textbf{s.t.}\ $d = d_b^{(j)}$\\
        && $M \lassign c_b^{(\hat j)} \oplus K$\\
        \hline
    \end{tabular}}

    \caption{Transfer phase of \IOT divided into offline and online operations.}
    \label{fig:transfer}
\end{figure*}

\subsubsection{CheckReveal phase}

The check-reveal phase corresponds to breaking down the scheme to ensure that the protocol was run fairly.

It consists of the Sender revealing $p$ and $q$ to the Receiver.
This clearly allows the Receiver to evaluate all the roots of $r$ and $-r$ and hence to verify that messages $M_0$ and $M_1$ were actually different, as they can be decrypted from the corresponding $\{c_i^{(j)}\}$ with ease.

It is interesting to notice that Receiver-side, obliviousness is maintained in this scenario (and hence also in the case of a compromising of the system, see later the upcoming Section~\ref{ssec:securityAnalysis}).
If the context requires the Receiver's choice to be provable too, then the Receiver can reveal $k$ as well, so that the Sender can determine which $\hat i$ is such that $k \in \{k_{\hat i}^j\}_{j\in\{0, 1\}}$; in this case, $k$ should be communicated prior to the factorisation of $n$ in order to be trustworthy (otherwise the Receiver can evaluate the other roots with ease).

Also, it is notable that the choice reveal does not require the Sender to interact with the Receiver, hence requiring no disregarding of the entire system; a similar outcome can be achieved also on the Sender side, actually mounting a partial check-reveal that allows the Sender provably reveal both messages while keeping the system intact.
To do so, a simple yet effective way is, for example, to let the Sender generate two more message nonces $s_0$ and $s_1$, and to send the Receiver the digests $\hat M_i = H(M_i \oplus s_i)$ along with the ciphertexts; this allows the Receiver to verify $M_{1-b}$ if $M_{1-b}$ and $s_{1-b}$ are revealed, while not compromising the obliviousness of $M_{1-b}$, even in the case of short messages (\ie when plaintext space is small enough to allow a brute force search).


\subsection{Transfer complexity analysis}
\label{ssec:complexity-analysis}

In this section, we analyse the computational complexity of the proposed \IOT protocol in terms of the key size $\lambda$.

In particular, all the operations (but $H(\cdot), F_s(\cdot)$ and $\oplus$ evaluations) are performed in $\mathbb{Z}_n$ where, typically, optimal algorithms do exist that execute either in $\OOO(\log n)$ or $\OOO(\log p + \log q)$ bit operations, with $n = p\cdot q$.
Here we can recall that $\log n \sim \lambda$ and, due to the typical cryptographic constraints on $p$ and $q$, we can claim $\log p \sim \sfrac\lambda2 \sim \log q$, hence bounding most operations (including also bit-wise xor) to $\OOO(\lambda)$.

In what follows, we separately assess the operations executed by the \emph{Receiver} during transfer (online) and prior to engage it (offline) and by the \emph{Sender}, as well as the relative communication cost.

\subsubsection{Receiver-side complexity}

The Receiver's overhead is purposely kept minimal to accommodate resource-constrained devices as it is common in IoT applications.
For this reason, we design simpler operations on the Receiver side, a complexity analysis of which is detailed in what follows and summarised in Table \ref{tab:receiver_cost}.
\begin{description}
    \item[Random choices] The Receiver is required to offline uniformly sample (or choose) a choice bit $b$ and a positive ring element $k \sample \ZZ_{\sfrac{n}{2}}$, accounting for $\OOO(\lambda)$ random bit generations. Furthermore, ensuring $\gcd(k, n) = 1$ and $k > \sqrt{n}$ can also be carried out within $\OOO(\lambda)$ operations in the worst case.

    \item[Modular squaring] To evaluate the quadratic residue $r = (-1)^b \cdot k^2 \bmod n$, the Receiver performs one squaring and one additive inversion modulo $n$, summing up to one addition ($n-t$ if $b = 1$) and one multiplication which can be carried out with $\OOO(\lambda^2)$ bit operations under classical algorithms. Squaring can be performed offline, while inversion depends on the individual interaction (yet offering nearly negligible overhead).

    \item[Message decryption] In order to determine the correct ciphertext to decrypt, the Receiver must perform a single hash evaluation $H(k)$, typically requiring $\OOO(\lambda)$ bit operations, which can be pre-evaluated.
    Then, at transfer time, a single evaluation of $F_s(\cdot)$ and a single xor are sufficient to retrieve $M_b$.
    Here, evaluating $F_s(\cdot)$ clearly dominates the online phase, potentially leading to unexpected increases in complexity.
    However, different solutions do exist to keep the computational burden low, while maintaining the essential properties required for securely and correctly run the protocol; in our implementation (see later Section~\ref{sec:experiments}), \eg, we considered the sponge hash function SHA-3 \cite{SHA-3} in its XOF form (namely SHAKE-256/512), actually instantiating $F_s(k) = H_3(s+k)$, where $s$ is chosen in $\ZZ_n$ for simplicity and $H_3$ represent SHAKE-256.
    This keeps the computational burden bounded to $\OOO(\lambda)$\footnote{Here, to be precise, computational cost is $\OOO(2\lambda\cdot \gamma^e)$, for some $e \in (1, 2]$, where $\gamma$ represent the required security level of the hash function under collision, namely an attacker is required $\OOO(2^\gamma)$ attempts to find a collision.}.
\end{description}

\begin{table}[!t]
\caption{Transfer phase (Offline and Online) computational cost in bit-operations (Receiver-side).}
\label{tab:receiver_cost}
    \footnotesize
    \centering
    \def\arraystretch{1.2}
    \setlength{\tabcolsep}{8pt}
    \begin{tabular}{c|c|c|c}
        \toprule
        \textbf{Operation} & \textbf{Cost per op} & \textbf{Number of op} & \textbf{Total cost} \\
        \midrule
        \multicolumn{4}{c}{Pre-computation (Offline)}\\
        \midrule
        Bit generation & $\OOO(1)$ & $\OOO(\lambda)$ & $\OOO(\lambda)$\\
        Greatest common divisor & $\mathcal{O}(\lambda)$ & 1 & $\mathcal{O}(\lambda)$ \\
        \textbf{Modular multiplication} & $\mathcal{O}(\lambda^2)$ & 1 & $\mathcal{O}(\lambda^2)$ \\
        Hashing & $\mathcal{O}(\lambda)$ & 1 & $\mathcal{O}(\lambda)$ \\
        \midrule
        \multicolumn{4}{c}{Transfer (Online)}\\
        \midrule
        Bit-wise xor & $\OOO(\lambda)$ & 1--2 & $\OOO(\lambda)$\\
        KDF & $\mathcal{O}(\lambda)^\star$ & 1 & $\mathcal{O}(\lambda)^\star$ \\
        \bottomrule
    \end{tabular}

    \smallskip
    \noindent$^\star$ depending on the design/implementation of $F_s(\cdot)$.
\end{table}

\subsubsection{Sender-side complexity}

Conversely to the Receiver, the Sender is assumed to be computationally capable of executing more intense operations with ease, namely executing the Tonelli-Shanks algorithm for evaluating QR root sets and applying the CRT for combining said roots.
The complexity of each individual function is summarised in Table \ref{tab:sender_cost} while an in-depth description follows:

\begin{description}
    \item[Square roots evaluation] Exploiting the knowledge of $p$ and $q$, the Sender can evaluate the roots of a QR by evaluating the roots individually in $\ZZ_p$ and $\ZZ_q$ and later combining them according to the CRT.
    In particular, the four roots of $r$ can be evaluated by applying the Tonelli--Shanks algorithm, hence resulting in a computational cost of $\mathcal{O}(\log p)) \sim \OOO(\sfrac\lambda2)$ operations each, while combining roots of size $\lambda$ again costs $\OOO(\lambda)$ operations (mainly to evaluate the Extended Euclidean Algorithm).
    Conversely, evaluating the four roots of $(-r)$ benefits the pre-evaluation of $I$, actually allowing saving most of the computation (only 2 multiplications and at most two subtractions are needed).
    However, in both cases, it is important to notice that multiplications, divisions, and squarings are involved, actually prompting a total cost of $\OOO(4\lambda^3)$ operations for both cases.

    \item[Random choices] The Sender is required to uniformly sample the nonce $s$ (ring element), requiring $\OOO(\lambda)$ random bit generations. This is the only operation the Sender can anticipate in an offline phase.

    \item[ciphertext evaluations] Building ciphertexts requires the same set of operations already described for the Receiver decryption, all caveats of the case included. However, the Sender is required to process four different ciphertexts, including four evaluations of $F_s(\cdot)$, hence assessable as $\OOO(4\lambda)$ in total.

    \item[Digests evaluation] The Sender is required to evaluate the digests $d_i^{(j)}$ for the four plausible roots, actually incurring in four evaluations of $H(\cdot)$, \ie $\OOO(4\lambda)$ bit operations.
\end{description}

\begin{table}[!t]
\caption{Transfer phase (Offline and Online) computational cost in bit-operations (Sender-side).}
\label{tab:sender_cost}
    \centering
    \footnotesize
    \def\arraystretch{1.2}
    \setlength{\tabcolsep}{8pt}
    \begin{tabular}{c|c|c|c}
        \toprule
        \textbf{Operation} & \textbf{Cost per op} & \textbf{Number of op} & \textbf{Total cost} \\
        \midrule
        \multicolumn{4}{c}{Pre-computation (Offline)}\\
        \midrule
        Bit generation & $\OOO(1)$ & $\OOO(\lambda)$ & $\OOO(\lambda)$\\
        \midrule
        \multicolumn{4}{c}{Transfer (Online)}\\
        \midrule
        \textbf{Tonelli-Shanks} & $\mathcal{O}(\lambda^3)$ & 2 & $\mathcal{O}(2\lambda^3)$\\
        \textbf{CRT} & $\mathcal{O}(\lambda^3)$ & 4 & $\mathcal{O}(4\lambda^3)$ \\
        Modular multiplication & $\mathcal{O}(\lambda^2)$ & 2 & $\mathcal{O}(2\lambda^2)$ \\
        KDF & $\mathcal{O}(\lambda)^\star$ & 4 & $\OOO(4\lambda)^\star$ \\
        Bit-wise xor & $\mathcal{O}(\lambda)$ & 4 & $\OOO(4\lambda)$ \\
        Hashing & $\mathcal{O}(\lambda)$ & 4 & $\OOO(4\lambda)$ \\
        \bottomrule
    \end{tabular}

    \smallskip
    \noindent$^\star$ depending on the design/implementation of $F_s(\cdot)$.
\end{table}

\subsubsection{Communication cost}

In IoT applications, particularly in the oblivious transfer context, minimising communication costs is crucial to ensure energy efficiency, reduce latency, and enable reliable operation in resource-constrained or bandwidth-limited networks.
For this reason, the protocol was designed to reduce as much as possible the message exchange between the Sender and the Receiver.
In particular, here a single transfer requires the Receiver communicating the quadratic residue $r$ and the Sender responding with the KDF key $s$, both being elements of $\ZZ_n$ in our implementation.
Apart from this initial exchange, the other sole transfer consists of two ciphertexts per message along with the digest of the four corresponding keys.
If needed, digests communication can be circumvented by solutions like the one proposed in Footnote~\ref{fn:replaceHashWithMask};
however, in practical implementations, digests are typically small in size (say 256 bits) if compared to modulus $n$ (say at least 2048 bits), hence making a single ring element larger than the 4 digests themselves.
The total communication complexity can be therefore bound with 2 ring elements and 4 messages (encrypted).


\section{Security of \IOT}
\label{sec:security-analysis}

In this section, we discuss the security of the proposed protocol and its obliviousness (see Section~\ref{ssec:system-model}) w.r.t.\ the threat models introduced in Section~\ref{ssec:threat-models} and summarised in Table~\ref{tab:threat_models}.
In particular, we structure the discussion considering one attacking entity at a time, namely the semi-trusted Sender, the malicious Receiver, the passive attacker, and the active attacker.
Finally, we consider the case of a malicious Sender and we show how the Receiver can play a game to check whether the Sender is acting honestly.


\subsection{Security Analysis}\label{ssec:securityAnalysis}
We now proceed to discuss the security of our model against the threats introduced in the previous section.
In particular, we analyse the security of our solution from the perspective of one involved entity at a time.

\subsubsection{(s) Honest-but-curious Sender}
\label{ssec:sender-security}

We start by observing that breaking the obliviousness of the Receiver's choice would require the Sender to determine whether the value $r$ received during the transfer phase is of the form $r = k^2$ or $r = -k^2$ for some $k \in \ZZ_n$.
Since $r$ is the sole value communicated by the Receiver to the Sender, any such attack must rely exclusively on distinguishing quadratic residues of the form $k^2$ from those of the form $-k^2$.

However, since the modulus $n = pq$ is generated with $p \equiv q \equiv 1 \pmod 4$, it follows that $-1 \in \QR_n$.
Consequently, both $k^2$ and $-k^2$ are quadratic residues modulo $n$.
Moreover, since $\gcd(k,n)=1$, each of the values $\pm k^2$ admits exactly four square roots modulo $n$, two of which are positive.
In particular, all four positive roots $\{k_i^{(j)}\}$ of $\pm k^2$ are equally likely and induce the same multiset $\{k^2, -k^2\}$, making the Receiver's choice bit $b$ (that determines which element of the pair is transmitted) perfectly hidden.
Since $k$ is sampled uniformly at random from $\ZZ_n$, from the Sender's perspective the two possible choices for $b$ are therefore equiprobable.
We note that pathological cases in which $\gcd(k,n) \neq 1$ would reduce the number of square roots and could, in principle, leak information; however, such events occur with negligible probability and would otherwise imply an accidental factorisation of $n$.

These intuitions, which can be formalised via an indistinguishability game as presented in Appendix~\ref{app:security-games}, show that threat \textbf{(s-i)} is prevented by design and that choice obliviousness holds unconditionally, provided the setup phase is executed honestly.

\subsubsection{(c) Honest-but-curious Receiver(s)}
\label{ssec:receiver-security}

We now consider the security of the protocol against an honest-but-curious Receiver.
As such, the Receiver follows the protocol specification but attempts to extract additional information from the received messages, beyond the intended output $M_b$.

We recall that each transfer consists of the Receiver selecting a session key $k \in \ZZ_n$ and a choice bit $b \in \{0,1\}$, and transmitting the value $r = (-1)^b k^2 \bmod n$ to the Sender.
In response, the Sender provides four ciphertexts $\{c_i^{(j)}\}$, encrypted under four distinct message keys $\{K_i^{(j)}\}$ derived from the square roots $\{k_i^{(j)}\}$ of $\pm r$ via a keyed KDF $F_s(\cdot)$, where the KDF key $s \sample \ZZ_n$ is freshly and uniformly chosen by the Sender.
Additionally, the Sender provides digests $\{d_i^{(j)}\}$ that allow the Receiver to identify the unique index $\hat j$ such that $k_b^{(\hat j)} = k$, and hence to recover $M_b$.

Given knowledge of $k = k_b^{(\hat j)}$, computing the conjugate root $k_b^{(1-\hat j)}$ is known to be computationally equivalent to factoring $n$, since
\[
    \gcd\bigl(n, k_b^{(\hat j)} - k_b^{(1-\hat j)}\bigr)
\]
reveals a non-trivial factor of $n$ \cite{Rabin1981}.
Similarly, computing any root $k_{1-b}^{(j)}$ of $-r$ would imply evaluating $I = \sqrt{-1} \bmod n$, which is hard under the same assumption.
Therefore, deriving any additional root from $(k,r)$ is computationally infeasible.

Under the security properties of the KDF $F_s(\cdot)$, each key $K_i^{(j)}$ is computationally indistinguishable from random unless the corresponding root $k_i^{(j)}$ is known.
Consequently, all ciphertexts $c_i^{(j)}$ except the one indexed by $(b,\hat j)$ are indistinguishable from encryptions under random keys.
Furthermore, even the knowledge of the plaintext $M_b$, partial knowledge of the plaintext space, or access to a decrypting oracle (which would allow recovering all $K_i^{(j)}$) reveals no information about the underlying roots due to the pre-image resistance of $F_s(\cdot)$.
The same argument applies to the digests $\{d_i^{(j)}\}$, which, under standard hash-function security assumptions, do not enable recovery, testing, or correlation of unknown roots in the absence of the corresponding root.
Here, in fact, we stress that cryptographic hiding and key indistinguishability rely exclusively on the $F_s(\cdot)$, while the hash function $H$ is used solely as an identification mechanism.

These observations, which can be formalised via the message obliviousness game presented in Appendix~\ref{app:security-games}, guarantee primitive security \textbf{(c-i)} and message obliviousness \textbf{(c-ii)} for a single transfer.

If multiple transfers are considered under the same KDF key $s$, the security properties of $F_s(\cdot)$ guarantee that all additional ciphertexts and digests remain computationally indistinguishable from random for the Receiver.
If different KDF keys are used across transfers, the independence of the KDF family $\{F_u(\cdot)\}_{u\in\KKK}$ ensures that no combination of prior information yields a non-negligible advantage on any new transfer.

Consequently, even when allowing adaptive access to previous ciphertexts and plaintexts, an honest-but-curious Receiver cannot exploit prior queries to create valid ciphertexts or infer unknown roots. Generating a valid ciphertext would require either computing unknown roots from $k$ and $r$ (equivalent to factoring $n$) or breaking the KDF security. Therefore, the protocol achieves transfer-level non-malleable CCA security (tNM-CCA).

Finally, collusion of polynomially many Receivers \textbf{(c-iii)} does not provide any additional advantage beyond what a single Receiver could achieve by executing the same number of transfers individually.
When each Receiver receive independent messages, collusion yields no new information.
If identical messages are used across users, the system only reveals what is inherently observable from the choice selection and does not compromise the protocol's obliviousness guarantees.

\subsubsection{(c-u) Untrusted Receiver(s)}

Let us now consider the possible enhanced behaviours of a malicious Receiver.
As already stated for (c), a malicious Receiver can influence the protocol only slightly.
In fact, its active role involves the transfer of the sole value $r$, which can be: (i) random valid (\ie a value obtained by squaring a random number); (ii) chosen valid (\ie obtained from a chosen $k$); or, (iii) possibly invalid \ie $r$ random or chosen.
The first case corresponds to not modifying the scheme, hence there is nothing to prove as it reduces to the honest Receiver case already analysed.
In the second case, the Receiver can choose $k$ to try achieving an advantage amongst different transmissions, \eg sending multiple times the same $r$ or trying to exploit a combination of multiple $k$s; however, the Sender entrusts the random choice of the nonce $s$ which will be different for each transmission; as a consequence, for each transfer a given key $k$ is sent to a different randomly-picked value $K$,
ensuring session independence of the derived keys, so that combining multiple protocol executions yields no additional information.
Hence, no advantages can be gained this way as well.
Finally, if $r$ is chosen without regard of $k$ (either randomly or not), then the Receiver loses its ability to know a root, and then it is unable to decrypt both messages; moreover, with a probability of $\sfrac34$, a random value for $r$ yields a non-quadratic residue, making the Sender aware of the tampering.
This proves the ineffectiveness of malicious Receiver behaviours as well.

\subsubsection{(p) Passive attacker}

The security under a passive attacker is directly guaranteed by the previous security statements:
since the Sender learns strictly more information than a passive attacker, any attack by the latter would contradict Sender obliviousness.
In fact, with the sole knowledge of $r$, $\{c_i^{(j)}\}$, $\{d_i^{(j)}\}$ and $s$, the attacker is neither able to retrieve the roots of $r$ to decrypt the messages (due to the hardness of the square root problem), nor to decide on the choice of $b$ (as the more powerful/informed Sender is unable as well).
Furthermore, as discussed for the Receiver, the combination of transfers is not exploitable as well, hence granting security under the passive attacker \textbf{(p-i)}.

\subsubsection{(a) Active attacker}

As already discussed, Receiver authentication is not addressed by the protocol, as no mechanism is provided to bind a transfer to a specific Receiver identity.
Consequently, Receiver spoofing (\textbf{a-iii}) is outside the threat model and can only be mitigated by external authentication mechanisms.
Such mechanisms are orthogonal to the obliviousness guarantees of the protocol and can be realised, for instance, at the transport layer or via standard public-key authentication, without impacting the protocol design.

Concurrently, Sender spoofing during the transfer phase (\textbf{a-ii}) is computationally infeasible.
In order to produce a ciphertext that the Receiver would successfully accept (and hence decrypt), an attacker must derive the specific root $k$ selected by the Receiver in order to evaluate both $H(k)$ and $F_s(k)$.
However, given only $r$, the Receiver's key $k$ is computationally hidden among the four square roots of $\pm r$, and identifying it without the factorisation of $n$ is equivalent to factoring $n$ itself.
Therefore, forging valid Sender responses is as hard as breaking the underlying quadratic residuosity or factorisation assumptions.

Finally, consider an attacker acting as a man-in-the-middle, capable of intercepting, modifying, or dropping messages.
To succeed without detection, the attacker would need to alter protocol messages while preserving Receiver decryptability.
However, any modification of the value $r$ results in a mismatch between the Receiver's secret key $k$ and the roots used to generate ciphertexts, making successful decryption impossible.
As in the case of Sender spoofing, producing consistent ciphertexts would require recovering the Receiver's root $k$, which is computationally infeasible.
Therefore, the protocol resists active tampering and man-in-the-middle attacks at the transfer level, showing the security under \textbf{(a-i)} and concluding the discussion.


\subsection{Working with untrusted senders}\label{ssec:untrustedServer}

We conclude the security discussion by considering the case of an untrusted Sender.

In particular, for what concerns the transfer phase, behaving maliciously does not make the Sender achieve any advantage over the Receiver. In fact, the Receiver choices ($k$ and $b$) are independent of the Sender actions and occur prior to any interaction, hence being not misdirectable by the actions of the Sender.
Furthermore, even in the case of a protocol corruption during the generation of the ciphertexts/digests couples, the Sender gains no significant advantage apart from being able to send the same message twice; however, the latter behaviour is not related to the malicious execution of the protocol but rather on the malicious input itself and it is usually solved with the potential check reveal phase, as already discussed.
It is also rather clear that the Sender can generate a valid encryption for, say, $m_0$ and an invalid one for $m_1$; then, if $b=0$ the tampering will be transparent to the Receiver, while if $b=1$ the Receiver will be unable to decrypt $m_1$, hence potentially revealing $b$ to the Sender. However, this behaviour is commonly accepted in literature as the Receiver becomes aware of the malicious conduct with probability $\sfrac 12$.

Conversely, behaving maliciously during the setup phase leads to much more interesting outcomes.
In particular, choosing $p, q \equiv_4 3$ leads to catastrophic security results, allowing the Sender to break the obliviousness of the Receiver's choice with ease.
In fact, it is well known that if $p, q \not\equiv_4 1$ then ${-1} \not \in \QR$, hence $-k^2 \not\in \QR$.
As a consequence, the Sender is simply required to check whether $r \in \QR$ (and hence $b=0$) or $r \not\in \QR$ (hence $-r \in \QR$ and $b=1$) and can still produce a valid encryption for $M_b$; clearly no valid encryptions are allowed for $M_{1-b}$ in this case, but the Receiver remains unaware of the tampering until a Check Reveal is issued.

To mitigate this problem, a simple yet clearly sub-optimal solution is requiring a trusted third party (TTP) to generate (or at least validate) $p$ and $q$.
Despite being effective and completely resolving all the related security issues, this solution misrepresents the scope and nature of OT protocols, which are typically run in distributed environments where no trust is enforced by any entity.
Hence, we propose a second solution (TTP-free) relying on a probabilistic game that the Receiver can run to check whether the Sender generated a valid $n$ or not.
It is worth noticing that this solution is novel and interesting \textit{per se}, as it can be generalised with ease to test most modularity properties of the factors involved in the creation of $n$ (however, this is out of scope of the present work, as it deserves proper discussion on its own).
Furthermore, multiple receivers can potentially collaborate in the game to reduce the moderate burden required.
The game is presented in Figure~\ref{fig:checkcongruency} and described as follows under the name of \textit{Check Congruency}.

\begin{figure}[!t]
    \centering
    \def\arraystretch{1.5}
    \small
    \begin{tabular}{|L{5.5cm}C{2cm}L{5.5cm}|}
    	\hline
        \textbf{Sender} & \textbf{Public} & \textbf{Receiver}\\
        \hline
        &&$X \sample \ZZ_n^\ell$ \textit{s.t.} $\forall x \in X \ \gcd(x, n) = 1$\\
        &&\textbf{for} $x \in X$ \textbf{repeat:}\\
        &&$\vert$\phantom{xx}$c \sample \{1, 2\}$\\
        $y$&$\xmapsfrom[\phantom{xxx}]{y}$&$\vert$\phantom{xx}$y = -x^c \mod n$\\
        $r = \texttt{isQR}(y)$ & $\xmapsto[\phantom{xxx}]{r}$ & $\vert$\phantom{xx}$r$\\
        &&$\vert$\phantom{xx}\textbf{if} $c = 1$\textbf{:}\\
        &&$\vert$\phantom{xx}$\vert$\phantom{xx}\textbf{save} $r$ \textbf{in} $R$\\
        &&$\vert$\phantom{xx}\textbf{elif} $r=\texttt{T}$\textbf{:}\\
        &&$\vert$\phantom{xx}$\vert$\phantom{xx}\textbf{continue}\\
        &&$\vert$\phantom{xx}\textbf{else:}\\
        &&$\vert$\phantom{xx}$\vert$\phantom{xx}\textbf{Stop!}\\
        &&\textbf{check that} \phantom{xx}$R \sim \begin{cases}\texttt{T} & \textit{w.p.}\ \sfrac14\\\texttt{F} & \textit{w.p.}\ \sfrac34\end{cases}$\\[1em]
        \hline
    \end{tabular}
    \caption{Check-congruency phase of the oblivious transfer.}
    \label{fig:checkcongruency}
\end{figure}

\subsubsection{Check congruency phase}

During the Check Congruency phase (after setup, before transfers), the Receiver challenges the Sender to prove that $-1 \in \QR$ by asking to distinguish between negative squares $-x^2$, which should be quadratic residues with probability 1 (and are non-residues if congruency fails), and random numbers $x$, which can be either residue or not with probability respectively $\sfrac14$ and $\sfrac34$.

In detail, $\ell$ numbers are randomly picked and potentially squared with probability $\sfrac 12$ (or any other probability known to the Receiver solely), \ie a random variable exponent $c$ is sampled from $\{1, 2\}$.
Then, the resulting values are multiplied by $-1$ and transmitted to the Sender which is required to answer whether those numbers are quadratic residues or not.
For each value $x \in \ZZ_n$, if $x$ is squared ($c=2$), then a negative response yields a proof of guilty as $-x^2 \not \in \QR$ if and only if at least one factor of $n$, say $p$, is such that $p \not\equiv_4 1$ (a positive response gives no further hints and can be discarded); conversely, if $x$ was not squared ($c=1$), then the response is stored in a response set $R$ in order to check whether is distributed accordingly the expected distribution of $\sfrac34$ negative, $\sfrac14$ positive, \ie $R\sim$ Bernoulli($\sfrac14$) or
\[
    R \sim \begin{Bmatrix}
        \texttt{T} & \mbox{with prob} & \sfrac14\\
        \texttt{F} & \mbox{with prob} & \sfrac34
    \end{Bmatrix} \ .
\]

We proved (see Appendix~\ref{app:untrusted-sender}) that the closest response a malicious Sender can provide leads to $R \sim$ Bernoulli($\sfrac34$), actually prompting that the Receiver can decide whether the Sender is behaving maliciously with probability ($>1-10^{-6}$) when $|R|>75$ (\ie $\ell \sim 215$) by simply checking whether $|\{r = \texttt{T} \mid r \in R\}| < \sfrac \ell2$.
The interested reader can find in Appendix~\ref{app:untrusted-sender} a complete discussion on this topic, including the proof for the best strategy a malicious Sender can apply and all the probabilistic estimates formalised under both classical and Bayesian approaches.

It is worth noticing that the computational cost of a single test is even lighter than a transfer since it does not require evaluating either hashes or square roots.
From the Receiver's perspective, the check congruency phase only involves $\ell$ random generation and (on average) $\sfrac \ell2$ modular squarings, making it feasible even with a few hundred tests.
Furthermore, communication-wise, the entire phase only requires the transmission of $\ell$ elements of $\ZZ_n$ and $\ell$ booleans that can be packed together with ease.
Finally, it is also possible for mutually trusted receivers to collaborate in generating and sending test values and hence assessing the trustworthiness of the Sender.


\section{\IOT Proof-of-concept}
\label{sec:experiments}

In this section, we present an empirical evaluation of the proposed \IOT protocol. We first outline the hardware and software choices, as well as the network conditions under which our proof-of-concept (PoC) implementation was tested. We then report the experimental results, focusing on protocol performance, correctness, and suitability for constrained IoT environments.


\subsection{The experimental setting}
\label{ssec:experimental-setting}

Our experimental setup involves two devices: an IoT embedded platform and a desktop general-purpose system. For the IoT configuration, we executed the protocol on a Raspberry Pi Zero 2W, equipped with a 1.0 GHz CPU and 512 MB of RAM, running the standard 32-bit Raspbian operating system.

For the desktop configuration, we employed a Dell XPS 15 9530 laptop running Ubuntu 24.04.1 LTS. The system was equipped with an Intel Core i9-13900H processor clocked at $\sim 3$ Ghz, 32 GB of RAM, and a 2 TB SSD. To ensure stable and reproducible performance measurements, the laptop remained connected to a power supply throughout the experiments and was not subjected to any additional compute-intensive workloads.

We conducted two separate experiments, one per device, in which both the Sender and the Receiver defined by our protocol were executed on the same machine. The two entities communicated logically in a simulated Receiver--Sender setting. Network traffic was not routed through an external router, as our evaluation focuses exclusively on protocol-level performance, specifically execution time and CPU cycle consumption.


\subsection{Results}
\label{ssec:experimental-results}

\begin{table*}[p]
    \centering
    \caption{Performance comparison (timings and cycles) between IoT- and Desktop-like devices across various key sizes.}
    \label{tab:results}
    \footnotesize
    \def\arraystretch{1.2}
    \setlength{\tabcolsep}{4.5pt}
    \rotatebox{90}{
        \begin{tabular}{c|l||cccc|cccc|cccc}
        \toprule
        \multicolumn{2}{c||}{} & \multicolumn{4}{c|}{Timings Desktop [\textmu{}s]} & \multicolumn{4}{c|}{Timings IoT [\textmu{}s]} & \multicolumn{4}{c}{Cycles [kilocycles, eval on Rpi]} \\
        \multicolumn{2}{c||}{\textbf{Key Size} $\lambda$} & 1024 & 2048 & 3072 & 4096 & 1024 & 2048 & 3072 & 4096 & 1024 & 2048 & 3072 & 4096 \\
        \multicolumn{2}{c||}{\textbf{Security} $\kappa$} & 80 & 112 & 128 & 150$^{\ref{fn:security-4096}}$ & 80 & 112 & 128 & 150$^{\ref{fn:security-4096}}$ & 80 & 112 & 128 & 150$^{\ref{fn:security-4096}}$ \\
        \midrule
        \midrule
        \multirow{5}{*}{ \rotatebox{90}{ \makecell{ \textbf{Receiver}\\\textbf{offline} } } }
         & mean & 3.38 & 6.26 & 10.32 & 15.95 & 91.40 & 145.00 & 214.63 & 283.93 & 91.19 & 144.79 & 214.37 & 283.71 \\
         & median & 2.89 & 5.68 & 9.71 & 15.49 & 89.14 & 140.86 & 205.91 & 276.04 & 88.96 & 140.60 & 205.68 & 275.78 \\
         & max & 10.67 & 16.88 & 22.63 & 31.41 & 159.22 & 202.92 & 345.62 & 365.83 & 159.06 & 202.76 & 345.37 & 365.57 \\
         & min & 2.43 & 4.21 & 6.93 & 9.59 & 70.83 & 119.43 & 193.12 & 268.12 & 70.62 & 119.27 & 192.97 & 267.92 \\
         & std & 1.45 & 2.12 & 2.80 & 3.69 & 15.73 & 13.96 & 23.63 & 18.81 & 15.73 & 13.96 & 23.64 & 18.81 \\
        \midrule
        \multirow{5}{*}{ \rotatebox{90}{ \makecell{ \textbf{Receiver}\\\textbf{online} } } }
         & mean & 1.05 & 1.90 & 2.80 & 3.83 & 21.70 & 30.43 & 39.90 & 49.37 & 21.50 & 30.22 & 39.69 & 49.14 \\
         & median & 1.03 & 1.88 & 2.78 & 3.82 & 21.51 & 30.42 & 39.77 & 49.17 & 21.30 & 30.16 & 39.58 & 48.96 \\
         & max & 1.58 & 2.27 & 3.32 & 4.39 & 36.20 & 31.93 & 43.33 & 51.88 & 35.99 & 31.72 & 43.12 & 51.67 \\
         & min & 0.96 & 1.77 & 2.60 & 3.51 & 21.04 & 29.95 & 39.11 & 48.54 & 20.83 & 29.74 & 38.91 & 48.33 \\
         & std & 0.09 & 0.08 & 0.11 & 0.19 & 1.53 & 0.34 & 0.68 & 0.67 & 1.53 & 0.34 & 0.68 & 0.67 \\
        \midrule
        \multirow{5}{*}{ \rotatebox{90}{ \makecell{ \textbf{Sender} } } }
         & mean & 252.84 & 967.49 & 2350.93 & 5026.96 & 5502.52 & 2.08e+04 & 5.49e+04 & 1.14e+05 & 5502.21 & 2.08e+04 & 5.49e+04 & 1.14e+05 \\
         & median & 254.56 & 963.79 & 2398.10 & 5064.36 & 5486.07 & 2.10e+04 & 5.50e+04 & 1.15e+05 & 5485.81 & 2.10e+04 & 5.50e+04 & 1.15e+05 \\
         & max & 295.50 & 1202.89 & 2961.89 & 6635.94 & 6411.25 & 2.57e+04 & 6.85e+04 & 1.44e+05 & 6410.83 & 2.57e+04 & 6.85e+04 & 1.44e+05 \\
         & min & 206.98 & 779.61 & 1882.48 & 3734.19 & 4612.29 & 1.67e+04 & 4.19e+04 & 8.61e+04 & 4611.98 & 1.67e+04 & 4.19e+04 & 8.61e+04 \\
         & std & 24.24 & 126.01 & 346.33 & 827.79 & 549.78 & 2757.27 & 9148.81 & 2.03e+04 & 549.78 & 2757.27 & 9148.81 & 2.03e+04 \\
        \midrule
        \midrule
        \multirow{5}{*}{ \rotatebox{90}{ \makecell{ \textbf{Total}\\\textbf{transfer} } } }
         & mean & 258.45 & 977.28 & 2366.16 & 5049.07 & 5645.85 & 2.10e+04 & 5.52e+04 & 1.15e+05 & 5645.57 & 2.10e+04 & 5.52e+04 & 1.15e+05 \\
         & median & 260.10 & 973.34 & 2413.22 & 5081.51 & 5620.96 & 2.12e+04 & 5.53e+04 & 1.15e+05 & 5620.70 & 2.12e+04 & 5.53e+04 & 1.15e+05 \\
         & max & 300.44 & 1214.02 & 2976.93 & 6656.60 & 6540.52 & 2.59e+04 & 6.88e+04 & 1.45e+05 & 6540.21 & 2.59e+04 & 6.88e+04 & 1.45e+05 \\
         & min & 211.00 & 789.23 & 1897.70 & 3755.69 & 4740.26 & 1.69e+04 & 4.22e+04 & 8.65e+04 & 4739.95 & 1.69e+04 & 4.22e+04 & 8.65e+04 \\
         & std & 24.59 & 125.76 & 345.09 & 827.11 & 544.74 & 2760.59 & 9142.22 & 2.03e+04 & 544.74 & 2760.58 & 9142.22 & 2.03e+04 \\
        \bottomrule
        \end{tabular}
    }
\end{table*}

In what follows, we discuss the results obtained during our test campaign.
To recreate reasonably realistic conditions, we configured the Sender part to generate prime factors as RSA keys using OpenSSL.
This ensures that typical security requirements are met and that the results are obtained using realistic modulus sizes and values.
Both Python and C implementations of OpenSSL allowed us to run tests for $\lambda \in \{1024, 2048, 3072, 4096\}$ corresponding to $\kappa \in \{80, 82, 128, 150\footnote{The estimation between 4096-bit RSA and 150-bit of symmetric security is proposed for comparison in \cite{nist80057pt1comments2005}. We report it for comparability as 4096-bit RSA keys are still widely adopted in practice.\label{fn:security-4096}}\}$, as values below 1024 are no longer considered secure and are therefore not supported.

For each transfer, we independently generated two random messages of size $\lambda$ (elements of $\mathbb{Z}_n$, for ease of simplicity) on the Sender side and a choice bit on the Receiver side, then executed the benchmark synchronously.

We considered the processing time as offline for the operations that the two parties can pre-compute, and online for the operations that shall be computed during active transfer (cf.\ Figure~\ref{fig:transfer}).

The average Sender-side processing time was primarily dominated by the composite square-root operation required by the protocol, whereas the Receiver overhead remained minimal. The latter mainly consisted of a single modular squaring operation --that can be computed offline-- and a one hash computation during the online phase (see Section~\ref{ssec:complexity-analysis}).
As expected, in all tests, the selected messages were always correctly retrieved.

As Table~\ref{tab:results} reports, the results from our experiments show that \IOT performs efficiently, with an average Receiver-side processing time of 39.90 \textmu{}s and 214.63 \textmu{}s respectively during the online and the offline phase to transfer a message using a 3072-bit key ($\kappa=128$) on the Raspberry Pi Zero~2W.
Furthermore, it can be observed that the computational load is overwhelmingly handled by the Sender, accounting for more than $99\%$ of the total processing time.
As the key size increases, the Sender's required processing time grows, yet it remains consistently below 7 ms, reaching a maximum peak of 6.85 ms in the worst-case scenario.
In contrast, the Receiver's processing time is only marginally affected by variations in the key size. On the Desktop device, the average computation time required to complete the online phase for $\lambda \in {1024, 2048, 3072, 4096}$ is 1.05, 1.90, 2.80, and 3.83 $\mu$s, respectively. The Receiver remains efficient even on IoT hardware, exhibiting computation times of 21.70, 30.43, 39.90, and 49.37 $\mu$s for the same set of operations. The results in terms of kilocycles (evaluated on the Raspberry Pi) are also reported in Figure~\ref{fig:receiver-kilocycles}.

\begin{figure}[!t]
    \centering
    \includegraphics[width=0.80\linewidth]{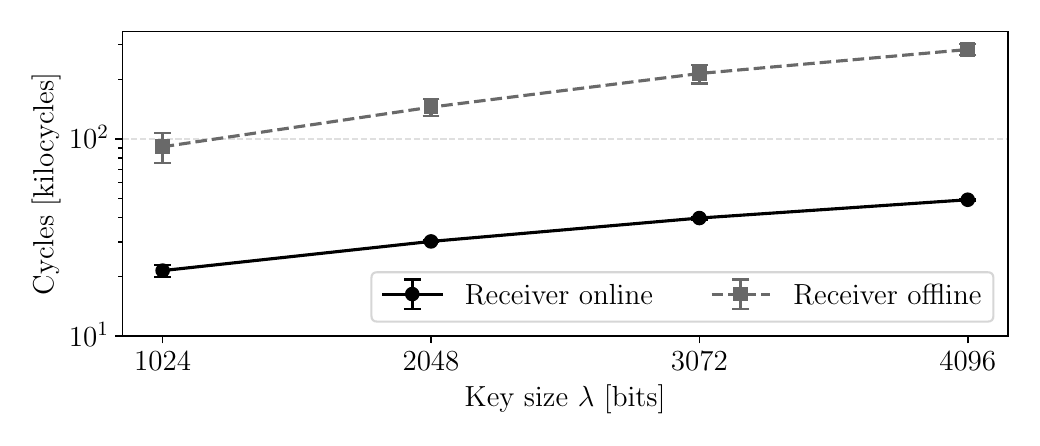}
    \caption{Performances (in kilocycles) achieved during receiver online and offline phases of the transfer.}
    \label{fig:receiver-kilocycles}
\end{figure}

In summary, considering the advantage of the pre-computation of the computationally expensive components required by the protocol --actually computed only once during the entire protocol life--regardless of the configuration, performance remains satisfactory for typical usage, confirming the feasibility of our \IOT protocol in low-power IoT environments.


\subsection{Comparison with SimplestOT}\label{ssec:simplestOT}

In order to provide a meaningful comparison with the state of the art, we decided to benchmark our approach against one of the most competitive (and widely adopted) references in the literature: SimplestOT protocol by Chou and Orlandi, originally presented in \cite{CO-OT15,CO-OT18}.
Also known as CO-OT, Simplest OT presents many similarities in design if compared to \IOT, due to its simplicity in definition (driven by Diffie-Hellman key exchange) and global efficiency (also due to the usage of ECC).
Furthermore, SimplestOT also might benefit from off-line pre-computation, despite not being exploited nor implemented in its original implementation: for this reason, we decided to make a porting of the original assembly code\footnote{The code is taken from \url{https://github.com/secretflow/simplest-ot} as the original link on the paper raise 404.} in C, implementing an online vs.\ offline benchmark.
Here, it is important to notice that our C porting clearly degraded the original performances claimed; however, results suggest our re-implementation remains competitive and valid, with a nearly negligible overhead of $2\times$ (cf.\ Table~\ref{tab:comparison_table}), expected under the change of language (as a reference, our Python implementation is three order of magnitude slower than the C implementation).

We present a detailed comparison of the two approaches in Figure~\ref{fig:timing-comparison}, which shows how \IOT protocol preserves superior efficiency on the Receiver's online phase, achieving 14$\times$ the performance on the IoT machine (25$\times$ resp.\ on the desktop machine).

Most notably, \IOT outperforms SimplestOT in the Receiver's online phase, even cross-platform, actually achieving lower execution time on the IoT device, even if compared with SimplestOT running on the desktop machine.
This result represents a significant step forward, as it clearly indicates that, in large-scale deployment scenarios, expensive, high-performance hardware can be potentially replaced with low-cost IoT devices, while still achieving better performance during OT execution.

Overall, these results validate the proposed design: the computational burden is effectively shifted to a well-resourced Sender --which can scale to higher performance depending on the setting-- ensuring lightweight requirements on the IoT Receiver.

\begin{figure}[!t]
    \centering
    \includegraphics[width=.80\linewidth]{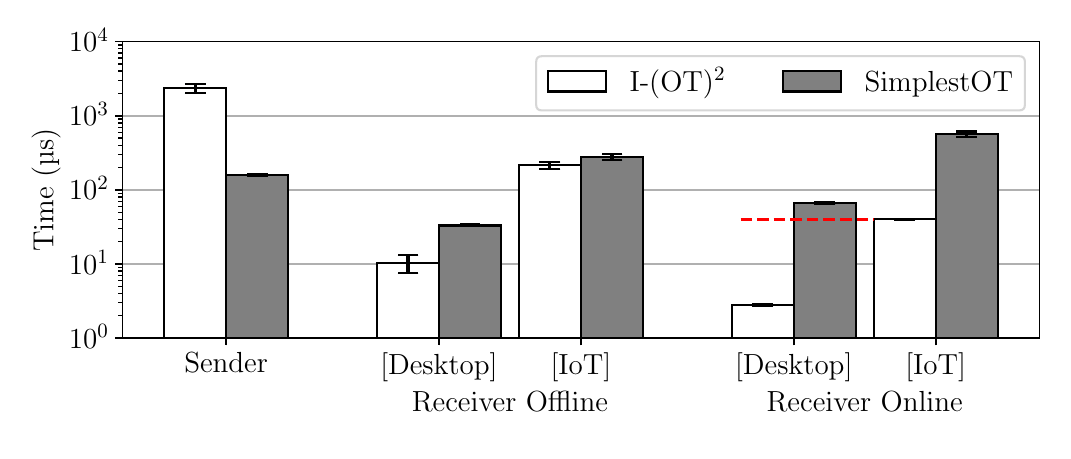}
    \caption{Performance (in \textmu{}s) comparison between our \IOT and our C porting of SimplestOT~\cite{CO-OT15}.}
    \label{fig:timing-comparison}
\end{figure}


\section{Conclusion and Future Work}
\label{sec:conclusions}

In this paper, we presented \IOT, a novel Oblivious Transfer protocol optimised for efficiency and security in client--server architectures, with a particular focus on resource-constrained environments such as IoT devices.
Our scheme substantially reduces Receiver-side computation and minimises communication overhead (achieving a $10\times$ performance improvement), while ensuring message confidentiality and non-malleability, thereby addressing key limitations of many existing OT constructions.
These design choices make \IOT particularly suited for contexts in which the number of required OTs is limited, and setup costs of extension protocols cannot be effectively amortised, favouring direct and lightweight base OT solutions.

We show the practical feasibility of \IOT through an open-source proof-of-concept implementation evaluated on both desktop- and IoT-like settings.
The results show extremely low receiver-side online computational cost, with execution times as low as 39.90~\textmu{}s for 128-bit security (3072-bit RSA modulus).
These findings confirm that \IOT is well-suited for real-world deployments requiring lightweight cryptographic operations, low interaction latency, and limited client capabilities.

Several directions remain open for further research.
First, while \IOT targets classical security settings, it is not secure
against quantum adversaries.
In this direction, we plan to explore post-quantum adaptations on lattice-based assumptions, such as (R-)LWE.
Second, extending the protocol to support batch or adaptive OT could further reduce overhead when multiple transfers are required.
Third, integrating \IOT within stronger composability frameworks, such as Universal Composability (UC)~\cite{Canetti2001UC}, would enable more comprehensive security guarantees.
Finally, at the implementation level, we aim at providing a more comprehensive framework to test our solution (and the major ones in the literature) in the context of extremely resource-constrained devices, like microcontrollers without OS support.


\bibliographystyle{elsarticle-num}
\bibliography{arXiv_biblio}


\appendix


\section{Probabilistic and Bayesian Analysis of the Check-Congruency Game}
\label{app:untrusted-sender}

In this appendix, we present a thorough probabilistic analysis of the Check-Congruency game introduced in Figure~\ref{fig:checkcongruency} and discussed in Section~\ref{ssec:untrustedServer}.
We start by reasoning on the Sender's optimal malicious strategy and we then derive the probabilistic estimates that justify the choice of $\ell$ in the protocol.
Finally, we conclude by presenting an alternative Bayesian formulation of the same problem that better captures the Receiver's confidence in the test result.

We recall the objective of the check congruency game being to allow the Receiver to verify whether the Sender generated a valid modulus $n$ (with $p\equiv_4 q \equiv_4 1$) or not, by challenging the Sender with random values and checking the distribution of the responses.
In details, during the Check-Congruency phase, the Receiver samples $\ell$ random values $x \in \ZZ_n$, which it then uses to challenge the Sender by sending $y = -x^c \bmod n$ for some random exponent $c \sample \{1, 2\}$.
The Sender is required to answer $r = \texttt{isQR}(y)$, \ie whether $y$ is a quadratic residue modulo $n$ or not.

Observe that, when the Sender is honest (\ie $p\equiv_4 q \equiv_4 1$), the response $r$ is fully determined by the protocol semantics.
Conversely, when the Sender is malicious, $r$ is no longer deterministic and can be partially controlled.
In this case, when $c = 2$, the value $-x^2$ is not a quadratic residue, but the Sender can attempt to deceive the Receiver by claiming it is.
However, when $c = 1$, the value $-x$ is a quadratic residue with probability $\sfrac14$ and the Sender cannot control the actual probability.

Let then $p_1, p_2 \in [0, 1]$ be the two probabilities $p_c = \PP[r = \texttt{T} \mid c]$ that the Sender produces in response to the Receiver's challenge.
The Receiver expects $p_1 = \sfrac14$ (the natural probability that a random element of $\ZZ_n$ is a quadratic residue) and $p_2 = 1$ (since the Sender knows the factorisation of $n$, it can always identify $-x^2$ as a quadratic residue).
The Sender, on the other hand, wants to deceive the Receiver by making $p_1$ as close as possible to $\sfrac14$ while also making $p_2$ as close as possible to $1$.
We now characterise the optimal strategy a malicious Sender can adopt, and the resulting values of $p_1, p_2$ that the Receiver observes.


\subsection{Best malicious sender strategy}

From a malicious Sender perspective (meaning $-1 \not\in\QR$) upon receiving $y = -x^c$, we have that
\[
    \PP[y \in \QR \mid c = 1] = \sfrac14
    \qquad,\qquad
    \PP[y \not\in \QR \mid c = 1] = \sfrac34
\]
and
\[
    \PP[y \in \QR \mid c = 2] = 0
    \qquad, \qquad
    \PP[y \not\in \QR \mid c = 2] = 1
    \ .
\]
Hence
\[
    \PP[y \in \QR] = \sfrac18
    \qquad, \qquad
    \PP[y \not\in \QR] = \sfrac78
    \ .
\]
As a consequence, it is straightforward to see that
\[
    \PP[c=1 \mid y \in \QR] = 1
    \qquad, \qquad
    \PP[c=2 \mid y \in \QR] = 0
\]
while, by applying Bayes theorem, we have
\[
    \PP[c = 1 \mid y \not\in \QR] = \frac 37
    \quad, \quad
    \PP[c = 2 \mid y \not\in \QR] = \frac47 \ .
\]
It follows that responding \texttt{T} with probability $\rho$ when $y \not \in \QR$ leads to a disclosure of the malevolent behaviour with probability $1 - \rho^{\sfrac47\ell}$.
It is easy to compute that, even for $\rho \sim 0.95$, the disclosure probability is higher than $\sfrac12$ as $\ell \geq 24$, and for $\rho > .95$ there is no sensible gain in risking the disclosure (later we estimate reasonable values of $\ell \sim 215$); hence the best strategy consists in setting $\rho = 1$, \ie responding \texttt{T} whenever $y \not \in \QR$.
It follows that, whenever $y \in \QR$, since $c = 1$ must hold, the response can safely be \texttt{F}.
This leads to the best strategy being
\[
    r = \begin{cases}
        \texttt{F} & \mbox{if } y \in \QR\\
        \texttt{T} & \mbox{otherwise}
    \end{cases} \ ,
\]
which yields $p_2 = \PP[r = \texttt{T} \mid c = 2] = 1$ and $p_1 = \PP[r = \texttt{T} \mid c = 1] = \sfrac34$, that is, $R$ being distributed like a Bernoulli($\sfrac34$) instead of the Bernoulli($\sfrac14$) expected from an honest Sender.

In what follows, we hence focus on the case $c = 1$, which is the only one that allows the Receiver to discriminate between honest and malicious Senders, since for $c=0$ the sender will always return a correct answer.
In particular, this can be done by checking the distribution of the responses against Bernoulli($\sfrac14$) $=D_0$ vs.\ Bernoulli($\sfrac34$) $=D_1$.

As usual, given $l$ consecutive i.i.d.\ samples, let $S_l$ be the number of positive ones sampled from the chosen distribution $D$.
Do note that here $l$ do represent the number of samples picked with $c=1$ and, since $c$ is uniformly sampled in $\{0,1\}$, then we typically require $\ell\sim 2\sqrt2\cdot l$ samples (it is easy to see that at least $l$ samples are picked with $c=1$ if $\ell = 2\sqrt2\cdot l$ when $l>40$).
Then, we have that
\[
    \PP[S_l = k \mid D_0] = \binom l k \left(\frac14\right)^k\left(\frac34\right)^{l-k}
\]
and
\[
    \PP[S_l = k \mid D_1] = \binom l k \left(\frac34\right)^k\left(\frac14\right)^{l-k}
\]
which can be approximated as a normal distribution $\mathcal N(\mu, \sigma^2)$ centred at $\mu_0=\sfrac l4$ (resp.\ $\mu_1 = \sfrac {3l}4$) with standard deviation $\sigma = \sqrt{\sfrac{3l}{16}}$ for large enough $l$;
a depiction of the corresponding PDFs for $l = 30$ and $l = 75$ can be found in Figure~\ref{fig:PDF-Sl} as a reference.

\begin{figure}[!t]
    \centering
    \includegraphics[width=.8\linewidth]{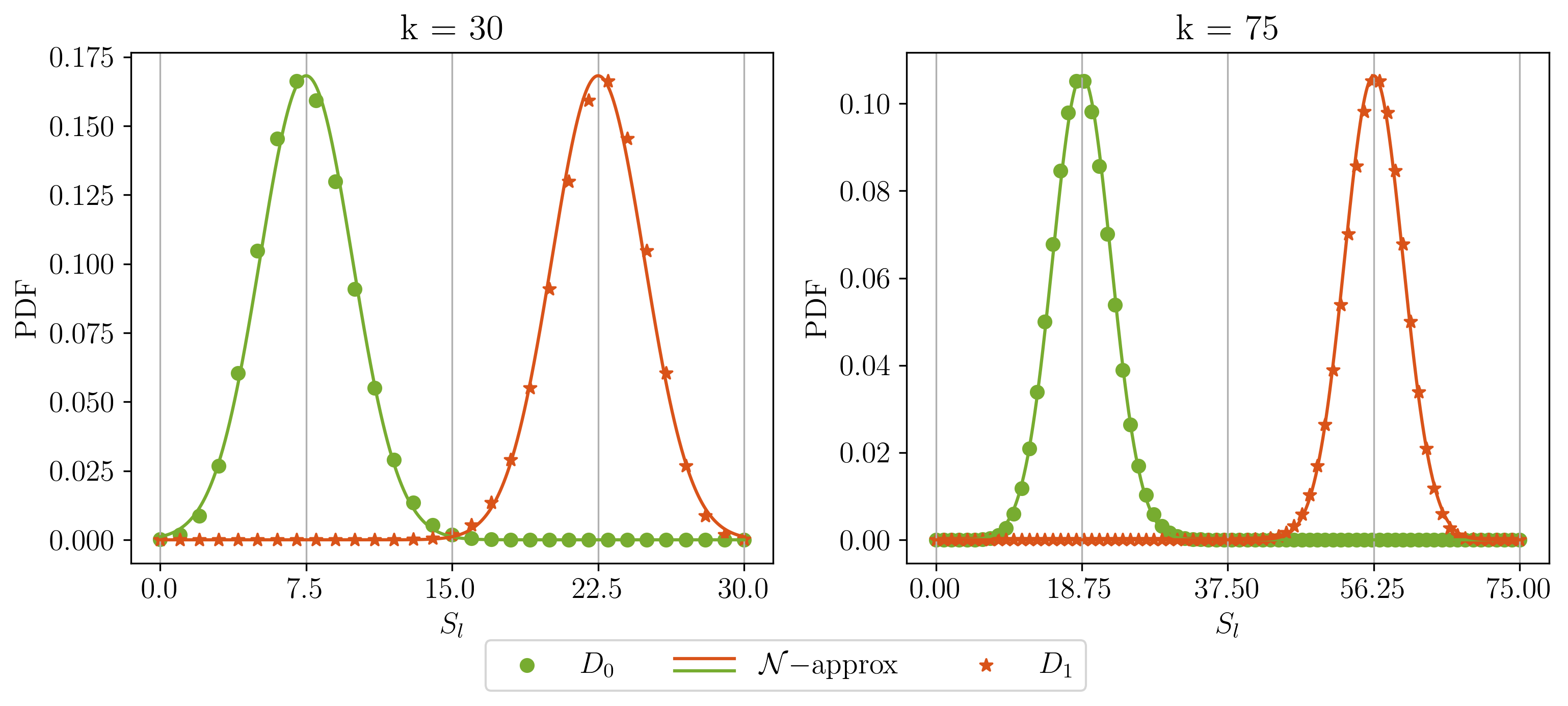}
    \caption{
        Probability Density Functions (PDFs) of elements sampled from the expected distribution $D_0$ (blue circles) and produced according to the best malicious strategy $D_1$ (orange stars) when the samples are (left panel) $k = 30$ and (right panel) $k=75$.
        The normal approximations of the two Bernoulli are represented for comparison as full lines.
        Do notice that $y$-scales are different being the domains different in size.
    }
    \label{fig:PDF-Sl}
\end{figure}

It can be noted that the two PDFs are centred in $\sfrac l4$ and $\sfrac{3l}4$ respectively, and are hence symmetrical one to the other w.r.t.\ $\sfrac l2$, suggesting that whenever $S_l > \sfrac l2$, then samples are drawn from $D_1$ with higher probability.
In particular, we have that
\[
    \PP[S_l > \sfrac l2 \mid D_0] = \sum_{k=\sfrac l2+1}^l \PP[S_l = k \mid D_0]
\]
or, by adopting the normal approximation, that
\[
    \PP[S_l > \sfrac l2 \mid D_0]
    = \PP\left[\frac{S_l - \mu}{\sigma} > \frac{\sfrac l2 - \mu}{\sigma} \right]
    = \PP\left[Z > \frac{\sfrac l4}{\sqrt{\sfrac{3l}{16}}} \right]
    = \PP\left[Z > \frac{\sqrt{l}}{\sqrt{3}} \right]
\]

with $Z$ being the standard normal distribution $\mathcal N(0, 1)$.
Hence, since $\PP[Z \geq 5] \approx 2.87 \times 10^{-7}$, we can confidently say that $\PP[S_l > \sfrac l2 \mid D_0]$ is negligible for $\sfrac{\sqrt{l}}{\sqrt{3}} \geq 5$, \ie $l \geq 75$.
Clearly, the same bound can also be obtained for the reverse proposition, meaning that $\PP[S_l < \sfrac l2 \mid D_1]$ is negligible when $l \geq 75$.
Collecting the two properties together, we can claim that ``the probability of the Sender deceiving the Receiver with confidence greater than $2.87 \times 10^{-7}$ is less than $2.87 \times 10^{-7}$ after $l = 75$ experiments''.
It follows that choosing $\ell = 2 \cdot \sqrt2 \cdot l \sim 215$ is way more than sufficient for the Receiver to be sure that the Sender is behaving honestly, see also Figure~\ref{fig:ERR-Sl}.

\begin{figure}[!t]
    \centering
    \includegraphics[width=.8\linewidth]{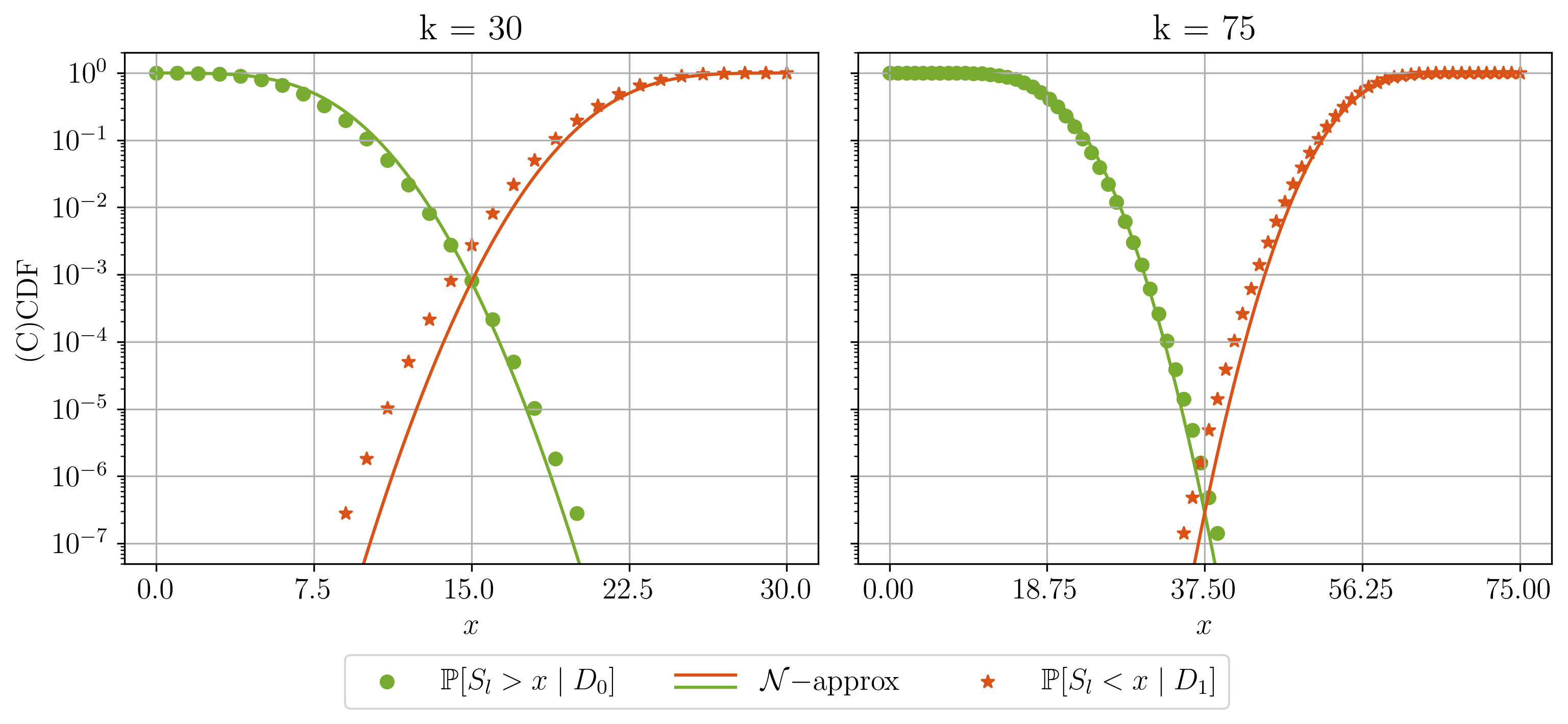}
    \caption{
        Probability of exceeding a given number $x$ of positive responses for the two distributions $D_0$ and $D_1$ from Figure~\ref{fig:PDF-Sl} when the number of samples is set to (left panel) $k = 30$ and (right panel) $k = 50$.
        $D_1$ (orange star) probabilities are modelled with the Cumulative Distribution Function (CDF), while values from $D_0$ (blue circles) are obtained from the survival function or Complementary Cumulative Distribution Function (CCDF).
        Normal approximations are represented as full lines for comparison.
        Do notice that perceivable (yet infinitesimal) numerical differences occur between exact and approximate values due to the inherent approximation of $D_0$ and $D_1$ being defined on discrete sets of values $x$.
    }
    \label{fig:ERR-Sl}
\end{figure}


\subsection{Bayesian approach to the effective size of the test set}

\begin{figure*}[!t]
    \centering
    \includegraphics[width=.99\linewidth]{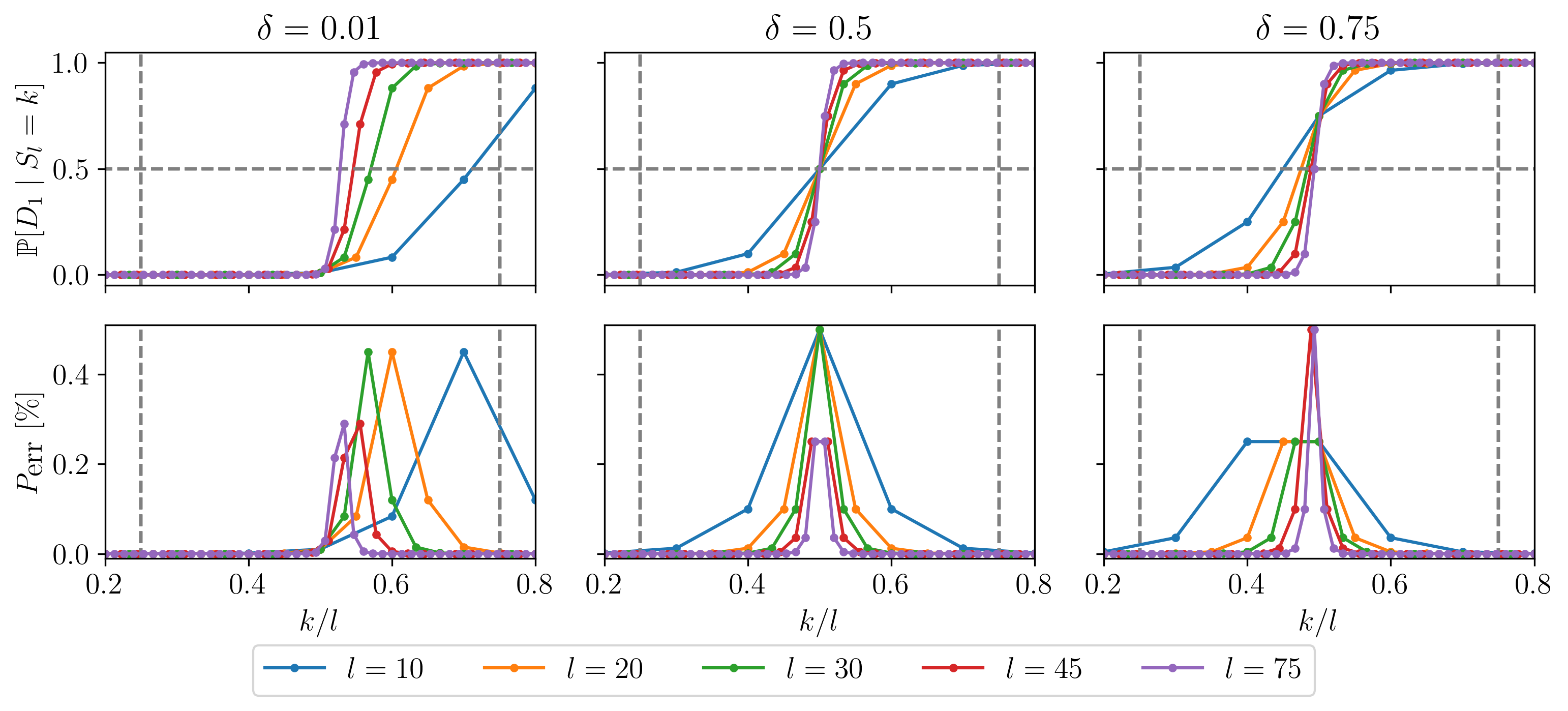}
    \caption{
        Bayesian analysis of the probability that the Sender is badly behaving w.r.t.\ the number of \textsc{t} samples obtained in $R$ (top panels) and corresponding uncertainty (bottom panels).
        Results are presented for three values of Receiver distrust $\delta \in \{\sfrac1{100}, \sfrac12, \sfrac34\}$ (from left to right).
        Do notice that the $x$-axis is 0-1 normalised to make plots at the varying of $l$ comparable and limited to $[\sfrac15, \sfrac45]$ for readability.
        Distribution means ($\sfrac12 \pm \sfrac14$) are highlighted by vertical dashed lines.
    }
    \label{fig:bayesian}
\end{figure*}

Alongside the classical probabilistic discussion, which involves the normal approximation of the Bernoulli and provides us with a suitable value for $\ell$, we can also discuss the confidence of any specific experiment result by adopting a Bayesian approach.
In particular, this gives us the opportunity to introduce a novel parameter $\delta \in [0, 1]$ representing the Receiver's distrust in the Sender's behaviour, where $\delta = 0$ implies the Sender is always trusted (honest-but-curious), $\delta = 1$ identifies a malevolent Sender who is always trying to cheat on the protocol, and $\delta=\sfrac12$ represents the unbiased situation where no \textit{a priori} information is available.
In the context of a Bayesian approach, we can shape such a value $\delta$ as the probability that $D_1$ was chosen rather than $D_0$, the so-called \textit{prior distribution}, \ie
\[
    \PP[D_0] = 1-\delta \qquad,\qquad \PP[D_1] = \delta\ .
\]

Given a specific experiment producing $k$ true responses out of $l$ tests, we can express the likelihood ratio of the Sender behaving maliciously as
\[
    \Lambda_l(k)
    = \frac{\PP[S_l = k \mid D_1]}{\PP[S_l = k \mid D_0]}
    = \frac{\binom lk \left(\frac34\right)^{k} \left(\frac14\right)^{l-k}}{\binom lk \left(\frac14\right)^{k}\left(\frac34\right)^{l-k}}
    = \left(\frac{\frac34}{\frac14}\right)^k
        \left(\frac{\frac14}{\frac34}\right)^{l-k}
        = \left(3\right)^k \left(\frac13\right)^{l-k}
        = 3^{2k-l}
\]
where it is interesting to notice that the log-likelihood ratio (LLR) is linear
\[
    \textsc{llr}_l(k) = \log(\Lambda_l(k))
    = \log(\PP[S_l = k \mid D_1]) - \log(\PP[S_l = k \mid D_0])
    = (2S_l - k) \log 3\ .
\]
Then, the probability that the Sender is behaving maliciously given the result of an experiment is:
\[
    \begin{split}
        \PP[D_1 | S_l = k] &
        = \frac{\PP[S_l = k \mid D_1] \cdot \PP[D_1]}{\PP[S_l = k]}
        = \frac{\PP[S_l = k \mid D_1]\cdot\delta}{\PP[S_l = k \mid D_1]\cdot\delta + \PP[S_l = k \mid D_0]\cdot(1-\delta)}
        \\&= \frac{\Lambda_l(k)\cdot\delta}{\Lambda_l(k)\cdot\delta + (1-\delta)}
        = \frac{1}{1 + \frac{(1-\delta)}{\delta}\Lambda_l(k)^{-1}}
        = \frac{1}{1 + \frac{(1-\delta)}{\delta}3^{l-2k}}
    \end{split}
\]
and analogously
\[
    \PP[D_0 \mid S_l = k] = \frac{1}{1+\frac{\delta}{1-\delta}\Lambda_l(k)} = \frac{1}{1+\frac{\delta}{1-\delta}3^{2k-l}}
\]
where we recall that $\frac{\delta}{1-\delta}3^{2k-l}$ is also referred to as posterior odds.

Finally, we can evaluate the uncertainty in the decision as the Bayesian decision error, \ie $P_\text{err} = \min(\PP[D_1 | S_l = k], \PP[D_0 | S_l = k])$, which evaluates to
\[
    P_\text{err} = \min\left(\frac{1}{1 + \frac{(1-\delta)}{\delta}3^{l-2k}}, \frac{1}{1+\frac{\delta}{1-\delta}3^{2k-l}}\right)
\]

Figure~\ref{fig:bayesian} reports the $\PP[D_1 | S_l = k]$ and the corresponding $P_\text{err}$ at the varying of the sample set size $|R| = l \in \{10, 20, 30, 45, 75\}$.
Results are presented for three example values of trustworthiness: strong Sender trust ($\delta = \sfrac{1}{100}$), no prior assumption ($\delta = \sfrac12)$, and small Sender distrust ($\delta = \sfrac34$).
In particular, do notice that, being the distribution discrete, then $P_\text{err}$ does not always reach the value of $\sfrac12$ as theoretically expected since it can happen that theoretical maximal uncertainty is located within two possible $k$ values: it is the case, \eg, when $l$ is odd and $\delta = \sfrac12$, as we can see for $l=75$; here the maximum is located in $\sfrac l2 = 37.5$, $\ie \sfrac kl = \sfrac12$, which is not an actual achievable number of \textsc{t} responses.


\section{Security Games}
\label{app:security-games}


\subsection{Choice Obliviousness Game}

The \emph{Choice Obliviousness Game} between a Receiver challenger $\CCC$ and a (possibly unbounded) Sender adversary $\AAA$ is defined as follows.

\begin{description}
    \item[Setup]
    The challenger generates an RSA modulus $n = pq$, where $p \equiv q \equiv 1 \pmod 4$, and sends $n$ to $\AAA$.

    \item[Query]
    $\AAA$ may request polynomially many honest executions of the protocol,
    obtaining the corresponding public values $r$ generated according to honest Receiver behaviour.
    These queries model full protocol transcripts available to the Sender.

    \item[Challenge]
    $\CCC$ samples a value $k \sample \ZZ_{\sfrac n2}$ and a challenge bit $b \sample \{0,1\}$ uniformly at random, and computes
    \(
        r = (-1)^b \cdot k^2 \bmod n \, .
    \)
    The value $r$ is given to $\AAA$.

    \item[Query]
    A second query phase is run as above, with fresh independent protocol executions.

    \item[Guess]
    The adversary $\AAA$ outputs a bit $b' \in \{0,1\}$.

    \item[Winning condition]
    The adversary wins if $b' = b$.
\end{description}

The advantage of $\AAA$ in the Choice Obliviousness Game is defined as
\[
    \Adv_{\mIOT}^{\mathsf{choice}}(\AAA)
    =
    \left| \Pr[b' = b] - \tfrac12 \right| \, .
\]

\begin{theorem}[Perfect Choice Obliviousness]
For any (possibly unbounded) Sender adversary $\AAA$, the advantage in the Choice Obliviousness Game is zero, \ie,
\[
    \Adv_{\mIOT}^{\mathsf{choice}}(\AAA) = 0 \, .
\]
\end{theorem}

\begin{proof}
    We begin the proof, noticing that the query phases provide the adversary only with independent samples drawn from the same distribution as the challenge value $r$, and therefore do not affect the indistinguishability argument below.

    For uniformly random $k \sample \ZZ_{\sfrac n2}$, the distributions of $k^2$ and $-k^2$ modulo $n$ are identical.
    Indeed, since $p \equiv q \equiv 1 \pmod 4$, there exists $I \in \ZZ_n$ such that $I^2 \equiv -1 \pmod n$, and thus
    \[
    (-1)\cdot k^2 \equiv (I \cdot k)^2 \pmod n \, .
    \]
    As multiplication by $I$ is a bijection over $\ZZ_n$, it follows that the random variables $k^2 \bmod n$ and $(I \cdot k)^2 \bmod n$ have identical distributions when $k$ is sampled uniformly from $\ZZ_{\sfrac n2}$.

    Therefore, the challenge value $r$ is identically distributed for $b=0$ and $b=1$.
    It follows that $\Pr[b' = b] = \tfrac12$ for any adversary $\AAA$, and hence
    \[
        \Adv_{\mIOT}^{\mathsf{choice}}(\AAA) = 0 \, .
    \]
\end{proof}


\subsection{Receiver Message Obliviousness Game}
Let $\AAA$ be a PPT honest-but-curious Receiver and let $\CCC$ be a challenger playing the role of the Sender.
We define the following game.

\begin{description}
    \item[Init]
    $\CCC$ generates and publishes the modulus $n = pq$, where $p \equiv q \equiv 1 \pmod{4}$.

    \item[Setup]
    $\CCC$ fixes the public description of the keyed KDF $F_\circ(\cdot)$ and the hash function $H(\cdot)$.

    \item[Query$_1$]
    $\AAA$ may adaptively request polynomially many honest executions of the protocol (under fresh $s$ values) by specifying arbitrary message pairs $(M_0, M_1)$.
    These queries model full protocol transcripts available to an honest-but-curious Receiver.

    \item[Challenge]
    $\AAA$ outputs two challenge messages $(M_0, M_1)$ of equal length, samples a fresh session key $k \in \ZZ_n$, selects a choice bit $b \in \{0,1\}$ (as it would regularly do it in the protocol), and sends $r = (-1)^{b} k^2 \bmod n$ to $\CCC$.

    The challenger samples a fresh KDF key $s^\star \sample \ZZ_n$ and sends it to $\AAA$.
    $\CCC$ then computes the four square roots $\{k_i^{(j)}\}$ of $\pm r$, derives the corresponding message keys $K_i^{(j)} = F_{s^\star}(k_i^{(j)})$, and forms the digests $\{d_i^{(j)}\}$ according to the protocol.

    For each pair $(i,j)$, the challenger samples an independent bit $\beta_i^{(j)} \sample \{0,1\}$ and replaces $K_i^{(j)}$ with a random value whenever $\beta_i^{(j)} = 0$, consequently creating and returning compliant or random cyphertexts to $\AAA$.

    \item[Query$_2$] A second query phase is run as before, under different fresh secrets $s \in \ZZ_n$.

    \item[Guess]
    $\AAA$ outputs guesses $\bar\beta_i^{(j)} \in \{0,1\}$ for all $(i,j) \neq (b,\hat j)$ and wins if
    \[
        \Pr\!\left[
        \bar\beta_i^{(j)} = \beta_i^{(j)}
        \right]
        \ge \frac{1}{2} + \varepsilon(\kappa)
        \quad
        \text{for some } (i,j) \neq (b,\hat j),
    \]
    where $\varepsilon(\kappa)$ is non-negligible.
\end{description}

\begin{theorem}[Receiver Message Security]

    Assume that factoring $n = pq$ is hard and that the keyed KDF $F_s(\cdot)$ and the Hash function $H(\cdot)$ are secure.
    Then, for any PPT honest-but-curious Receiver \(\AAA\), its advantage is defined as
    \[
    \Adv^{\mathsf{message}}_{\mIOT}(\AAA)
    = \max_{(i,j)\neq(b,\hat j)}
    \left|
    \Pr[\bar\beta_i^{(j)} = \beta_i^{(j)}] - \frac12
    \right|.
    \]
    We say that $\AAA$ wins if
    \[
        \Adv^{\mathsf{message}}_{\mIOT}(\AAA)
        \ge \varepsilon(\kappa)
    \]
    for some non-negligible $\varepsilon$.
\end{theorem}

\begin{proof}
    We first observe that the adversary's access to the two query phases does not increase its distinguishing power on the challenge execution.
    Each query execution is performed under a freshly sampled and independent KDF key $s$, and the Receiver learns at most one square root per execution.
    Under the hardness of factoring $n$, knowledge of a square root in one session does not enable computing any other square roots of $\pm r$ in that session, nor does it provide information about roots arising in independent sessions.

    Moreover, by the security of the keyed KDF $F_s(\cdot)$, for any query execution the distribution of ciphertexts corresponding to unknown roots is computationally indistinguishable from encryptions under random keys.
    Therefore, the entire view of $\AAA$ in both query phases can be simulated without knowledge of any additional square roots, and without affecting $\AAA$'s advantage in the challenge phase.

    Then, assume for contradiction that there exists a PPT adversary \(\AAA\) and a non-negligible function $\varepsilon(\kappa)$ such that
    \[
        \Adv^{\mathsf{message}}_{\mIOT}(\AAA) \ge \varepsilon(\kappa).
    \]

    Then there exists (at least) an index $(i^\star,j^\star) \neq (b,\hat j)$ such that \(\AAA\) distinguishes whether the ciphertext is real or random with non-negligible advantage $\varepsilon(\kappa)$.
    Under the security of the keyed KDF $F_{s^\star}(\cdot)$, the output $F_{s^\star}(x)$ is computationally indistinguishable from random for any unknown input $x$.
    Therefore, distinguishing a real key
    \[
        K_{i^\star}^{(j^\star)} = F_{s^\star}(k_{i^\star}^{(j^\star)})
    \]
    from random with non-negligible advantage implies recovering information about $k_{i^\star}^{(j^\star)}$.

    However, for any $(i^\star,j^\star) \neq (b,\hat j)$, the value $k_{i^\star}^{(j^\star)}$ is a square root of $\pm r$ distinct from the root $k$ held by the Receiver.
    Computing such a root from $r$ without knowing the factorisation of $n$ is equivalent to factoring $n = pq$, which is assumed to be computationally infeasible.

    Additionally, the digests $d_i^{(j)} = H(k_i^{(j)})$ provided to \(\AAA\) serve solely as identifiers and do not allow enforcing correlation tests under the security assumptions of the hash function:
    by the preimage-resistance property of $H$, they do not allow the Receiver to derive any unknown square roots \(k_i^{(j)}\) or correlate them with other roots, and therefore leak no information beyond what is already implied by $r$ (\ie check whether a candidate $c$ is a root, provided $\AAA$ can always verify it by evaluating $c^2=_{(?)}r$).

    Consequently, \(\AAA\) cannot distinguish $F_{s^\star}(k_{i^\star}^{(j^\star)})$ from a uniformly random value without either breaking the security of the KDF or factoring $n$, both of which contradict our assumptions.

    It follows that, for all $(i,j) \neq (b,\hat j)$,
    \[
        \left|
        \Pr\!\left[
            \bar\beta_i^{(j)} = \beta_i^{(j)}
        \right]
        - \frac{1}{2}
        \right|
        \le \negl(\kappa),
    \]
    and hence
    \[
        \Adv^{\mathsf{message}}_{\mIOT}(\AAA) \le \negl(\kappa).
    \]
    Since the above argument holds independently of the number and adaptivity of the query executions, it applies equally in the presence of both pre- and post-challenge queries.
\end{proof}

\vfill
\end{document}